\documentclass[a4paper]{article}
\usepackage[margin=20mm]{geometry}
\usepackage{listings}
\usepackage{amsmath}
\usepackage{amsthm}
\usepackage{caption,subcaption}
\usepackage {array}
\usepackage{mdwmath}
\usepackage{multirow}
\usepackage{mdwtab}
\usepackage{eqparbox}
\usepackage{multicol}
\usepackage{amsfonts}
\usepackage{multirow,bigstrut,threeparttable}
\usepackage{array}
\usepackage{bbm}
\usepackage{orcidlink}
\usepackage{epstopdf}
\usepackage{mdwmath}
\usepackage{mdwtab}
\usepackage{eqparbox}
\usepackage{tikz}
\usetikzlibrary{positioning, calc, shapes.geometric, shapes.multipart,
  shapes, arrows.meta, arrows,
  decorations.markings, external, trees, backgrounds,positioning,fit}
\usepackage{tikz-cd}

\usepackage{xr}
\usepackage{latexsym}
\usepackage{amssymb}
\usepackage{bm}
\usepackage{amssymb}
\usepackage{graphicx}
\usepackage{mathrsfs}
\usepackage{mathtools}
\usepackage{epsfig}
\usepackage{psfrag}
\usepackage{setspace}
\usepackage{algorithm}
\usepackage{algorithmic}
\usepackage{csquotes}

\usepackage[backend=bibtex,
citestyle=authoryear-comp,
bibstyle=authoryear-comp,
autocite=plain,
sorting=nyt,
hyperref=true,
maxcitenames=2,
maxbibnames=100,
isbn=false,
url=false,
doi=false]{biblatex}
\usepackage{hyperref}
\hypersetup{%
    pdfborder = {0 0 0}
}
\usepackage{cleveref}

\usepackage{soul}

\usepackage{soul}
\newcommand{\revc}[1]{{\color{black} #1}}
\newcommand{\revd}[1]{{\color{black} #1}}

\usepackage{preamble-tweak-2}

\title{Selective Randomization Inference for Adaptive Experiments}
\author{Tobias Freidling\orcidlink{0000-0003-0724-4297}\textsuperscript{1} \qquad Qingyuan Zhao\orcidlink{0000-0001-9902-2768}\textsuperscript{2}\footnote{Correspondence to: qyzhao@statslab.cam.ac.uk; Centre for Mathematical Sciences, Wilberforce Road, Cambridge, CB3 0WB, UK}\qquad Zijun Gao\orcidlink{0000-0003-4863-1656}\textsuperscript{3}}
\date{{\textsuperscript{1}Department of Mathematics, \'{E}cole polytechnique f\'{e}d\'{e}rale de Lausanne, Switzerland}\\
{\textsuperscript{2}Statistical Laboratory, DPMMS, University of
    Cambridge, United Kingdom}\\
  {\textsuperscript{3}Marshall School of Business, University of Southern California, USA}
}

\begin{document}
\def\spacingset#1{\renewcommand{\baselinestretch}%
{#1}\small\normalsize}
\spacingset{2}

\maketitle

\begin{abstract}
  Adaptive experiments use preliminary analyses of the data to inform further course of action and
  are commonly used in many disciplines including medical and social sciences. Because the null hypothesis and experimental design are \revc{data-dependent}, it has long been recognized that statistical inference for adaptive experiments is not straightforward. Most existing methods only apply to specific adaptive designs and rely on strong assumptions. In this work, we propose selective randomization inference as a general framework for analysing adaptive experiments. In a nutshell, our approach applies conditional post-selection inference to randomization tests. By using directed acyclic graphs to describe the data generating process, we derive a selective randomization p-value that controls the selective type-I error. \revc{As inference only relies on the randomness in the treatment assignment, no modelling assumptions or independent and identically distributed data are needed.}
  \revc{We elaborate on conditions that render the proposed p-value computable and provide rejection sampling and MCMC algorithms to find a Monte Carlo approximation. Moreover, this article shows how to estimate and construct confidence intervals for a homogeneous treatment effect.}
  Lastly, we demonstrate our method and compare it with other randomization tests using synthetic and real-world data.
\end{abstract}

\noindent%
response-adaptive randomization

\newpage
\section{Introduction}\label{sec:introduction}

In a talk in 1956 titled ``Iterative Experimentation'', the great
statistician George Box suggested that we should keep in mind that
scientific research is usually an iterative process. He wrote in the
abstract:
\begin{quote}
  The cycle: conjecture-design-experiment-analysis leads to a new cycle
  of conjecture-design-experiment-analysis and so on. [...] Although this
  cycle is repeated many times during an investigation, the experimental
  environment in which it is employed and the techniques appropriate for
  design and analysis tend to change as the investigation
  proceeds. \hfill \parencite{box1957}
\end{quote}
John Tukey, another great statistician and Box's contemporary,
acknowledged that Box's cycle ``will serve excellently'' if ``an
oversimple paradigm is to be selected'', but he noted that ``the point
where the circle is to be broken can be freely chosen'' and ``there
are short cuts and repeated steps as well as branchings and
rejoinings'' \parencite{tukey1961}. 
Thus, although the demarcation in
Box's cycle is very useful to understand the role of statistics in
different phases of experimentation, scientific research in the real
world does not always follow this linear process.

In practice, it is often desirable to continuously revise the
scientific hypothesis and design of the study as more data are gathered.
This has motivated two long-standing themes of research in statistics:
(i) \emph{selective inference}, where valid inference is required on potential
findings that are discovered after viewing the data, and (ii) \emph{adaptive
experimentation}, where data-adaptive designs are used to allow us to
collect more relevant data, save time and resources, and/or assign
more participants to a better treatment.

As the scale of datasets grows and increasingly precise conclusions are
sought, selective inference and adaptive experimentation have
attracted much attention in the recent literature. However, there is,
strikingly, little work on applying selective inference to adaptively
collected data. In this article, we attempt to fill this gap and
develop a general method called ``selective randomization inference''
that can be applied to a wide range of adaptive experiments. For the
rest of this section, we will introduce our proposal and then put it
in the broader context set up by Box and Tukey.

\subsection{Overview of the Selective Randomization Test}
\label{sec:overv-select-rand}

\begin{figure}[t]
  \centering
  \includegraphics[scale=0.185]{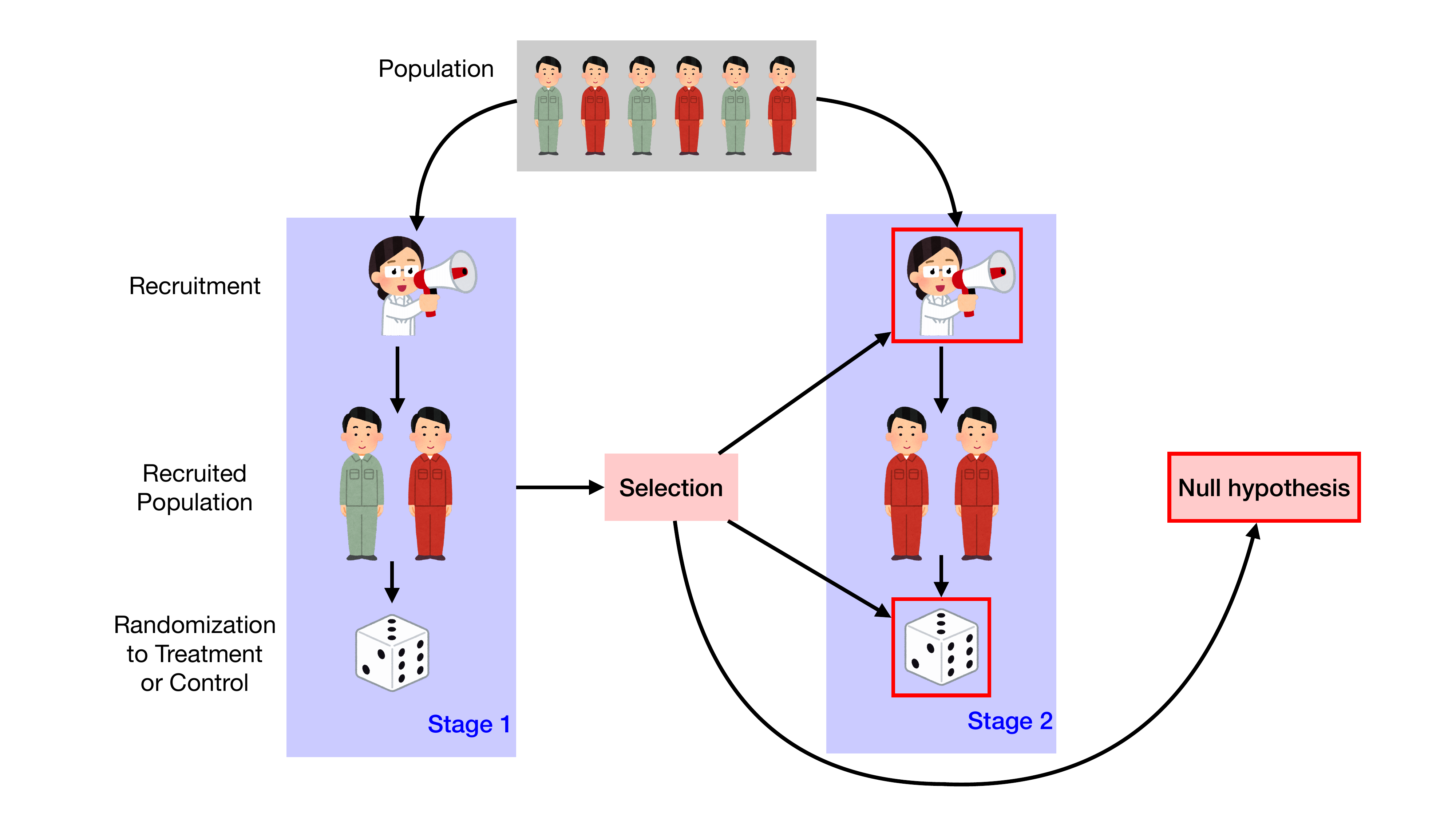}
  \caption{Schematic of a two-stage enrichment trial with two subgroups: low genetic risk score (green) and high genetic risk score (red).}
  \label{fig:trial-pictogram}
\end{figure}

To motivate the use of adaptive experiments, illustrate potential pitfalls, and explain our proposal, we introduce a hypothetical two-stage enrichment trial, which is a design that allows the experimenter to ``over-sample'' a subgroup of interest in the second stage. This running example is motivated by a post hoc analysis by
\textcite{marston_predicting_2020} of the FOURIER trial (Further
Cardiovascular Outcomes Research With PCSK9 Inhibition in Subjects With
Elevated Risk) that was brought to our attention by an applied
collaborator. In their work, \citeauthor{marston_predicting_2020} found
that the new evolocumab therapy (PCSK9 inhibition) appears to benefit patients
with high genetic risk much more than others. However, because of the
post hoc nature of this subgroup analysis, further randomized studies
are required to confirm this conclusion. Therefore, we consider an experiment with two stages that can adapt to possible effect heterogeneity; a graphical illustration of the example can be found in Figure~\ref{fig:trial-pictogram}.

\begin{example}\label{ex:enrichment}
    Consider a hypothetical two-stage enrichment trial in which the
    patient population is divided at the beginning of the study into a low
    genetic risk group and a high genetic risk group. In the first stage
    of the trial, patients are recruited from both groups and randomly
    assigned the new therapy ($Z_{1i}=1$) or the placebo ($Z_{1i}=0$,
    where $i$ is the patient index). Based on a preliminary analysis of
    data from the first stage, we may wish to select one subgroup or both
    subgroups for further investigation in the second stage. For
    example, we choose to only recruit people with high genetic risk
    score if they show more promising responses in the first stage and
    potentially assign the new therapy with higher probabilities as
    well. The treatment assignments in the second stage are denoted by $Z_2$. After concluding the experiment, we are interested
    in testing whether there is any treatment effect for the subgroup(s)
    selected at the end of the first stage, which is a categorical
    variable denoted by $S_1 \in \calS_1 = \{\text{only high}, \text{only low},
    \text{both}\}$.
\end{example}

It has been recognized very early on that a naive analysis of such a
sequential trial tends to overstate statistical significance
\parencite{armitage_sequential_1960, anscombe_sequential_1963,
  pocock_group_1977}. To be more concrete, consider the following
randomization p-value that completely ignores selection. Let $Z = (Z_1,Z_2)$
denote all the treatment assignments and~$W$ denote the recruitment
decisions, covariates and potential outcomes of the patients. A
randomization test conditions on~$W$ and only uses the randomness in
the treatment assignment~$Z$ \revd{for inference. In adaptive experiments, however, the corresponding randomization distribution $\Prob(Z=\cdot \mid W)$ is hardly tractable due to the dependence of future recruitment decisions on past treatments. Therefore, one may try to use the sequential treatment assignment distribution~$\QQ$  that conditions the treatment assignments only on past data (see \eqref{eq:A1} below for the precise definition)} to calculate the following p-value
\revd{
\begin{equation}\label{eq:naive-rand-p-value}
  P_{\text{naive}} = \QQ(\,T(Z^*, W) \leq T(Z,W) \mid Z, W\,).
\end{equation}}
Here, $T(Z,W)$ is a test statistic that measures treatment efficacy
and $Z^* = (Z^{*}_1, Z^{*}_2)$ 
\revd{is independently drawn from the same treatment assignment distribution as~$Z$.} 
Because $Z$ and $W$ are
treated as given in \eqref{eq:naive-rand-p-value}, the probability is
computed using the distribution of $Z^{*}$, an independent realization
of the treatment assignment. \revd{Due to subgroup selection, $\QQ(Z=\cdot \mid W)$ does not equal the randomization distribution $\Prob(Z=\cdot \mid W)$ and the null hypothesis is data-dependent. Hence,} 
$P_{\text{naive}}$ is generally not a valid p-value in the sense that \revc{under the null hypothesis}
$\Prob(P_{\text{naive}} \leq \alpha)$ may be much larger than the
chosen significance level $ \alpha \in [0,1]$. This is because even
though the
treatment assignment~$Z$ is randomized, the selection ``changes'' its
distribution. Heuristically, suppose the high genetic risk group is
selected at the end of stage 1 (so $S_1 = \text{only high}$), then the
first stage data must have already shown promises for the new therapy
in the high genetic risk group, and it is no longer appropriate to use
the randomization distribution of $Z_1$ to calculate the p-value. In
other words, it is ``unfair'' to compare $Z_1$ with $Z_1^*$ generated
from the original randomization distribution because $Z_1$ has gone
through selection but $Z_1^*$ does not.

A convenient solution to this problem is to discard the first stage,
aka data splitting \parencite{cox_note_1975}. That is, we use the data
from the first stage to choose the design and null hypothesis, and
then use the second stage only for statistical inference. Mathematically,
this can be expressed as using the following randomization p-value:
\revc{
\begin{equation}\label{eq:2nd-stage-p-value}
  P_{\text{split}} = \Prob(\,T(Z^*,W) \leq T(Z,W) \mid Z, W, Z_1^*=Z_1\,).
\end{equation}
Since we condition on $\{Z_1^*=Z_1\}$, $P_{\text{split}}$ only uses the randomness in the stage II treatment assignment $Z_2$} \revd{and the randomization distribution and $\QQ$ agree.} Obviously, this p-value is valid but is inefficient because the first stage data are completely discarded in the test.

In a nutshell, our proposal is to condition on the exact amount of
information used in the selection. This is inspired by the recent
literature on conditional post-selection inference and this connection
will be further explained in \Cref{sec:related-work}. The key observation is
that the selection in the first stage $S_1 = S_1(Z_1, W)$ is a function of the treatment assignment $Z_1$ \revc{and precisely describes the information, or rather ``randomness of $Z_1$'', that is used to choose the design of the second stage.}
Our main idea is that, to
appropriately adjust for selection, one should simply condition on the
selection and use the following \emph{selective randomization
  p-value}:
\begin{equation} \label{eq:selective-p-value}
  P_{\text{sel}}=\Prob(\,T(Z^*,W) \leq T(Z,W) \mid Z, W,
  S_1(Z_1^*, W)=S_1(Z_1, W)\,).
\end{equation}
We will show that this p-value \revd{is computationally tractable and} controls the \emph{selective type-I
  error} in the sense that
\begin{equation} \label{eq:selective-error}
  \Prob(P_{\text{sel}} \leq \alpha \mid W, S_1=s_1) \leq
  \alpha,\quad
  \text{for all}~s_1\in \calS_1~\text{and}~ \alpha \in [0,1].
\end{equation}
Thus, no matter which subgroup is selected at the end of the first
stage, the false positive rate of this selective randomization test
for the selected null hypothesis is controlled. Compared to data
splitting, this test is expected to be more powerful because it uses
all ``left out information'' in the first stage that is not used by
the selection.

\subsection{Related Work}
\label{sec:related-work}

Our proposal is related to a wide range of research that offers
alternatives to the linear cycle of scientific research described by
Box. Next we will review the related literature and compare it with
our proposal.

Adaptive designs can be traced back to the work of
\textcite{thompson_likelihood_1933} and are now frequently used in
many fields of empirical
research including medicine \parencite{burnett_adding_2020}, political
science \parencite{offer-westort_adaptive_2021}, and economics
\parencite{duflo_chapter_2007, kasy_adaptive_2021}. In many
situations, the experiments have several designated stages, and within
each stage the participants may be admitted on a rolling
basis. Therefore, preliminary data from the initial participants are
often available before the study concludes. This forms the basis for
adapting the design (e.g.\ recruitment rule or treatment allocation)
in later stages of the experiment. For instance, it may be desirable
to increase the
assignment probability of more promising treatments or policies
(this is called \emph{response-adaptive randomization}, see e.g.\
\textcite{rosenberger_covariate-adjusted_2001} and
\textcite{faseru_changing_2017}), discontinue certain arms of the study
early (this is called \emph{multi-arm multi-stage design}, see e.g.\
\textcite{magirr_generalized_2012} and \textcite{sydes_issues_2009}),
or focus on the subpopulation of participants that exhibit the best
reactions (this is called \emph{enrichment design}, see e.g.\
\textcite{magnusson_group_2013} and
\textcite{ho_efficacy_2012}). \textcite{pallmann_adaptive_2018}
provided an excellent review of adaptive designs in clinical trials
from a practical perspective.

Logistical hurdles and statistical challenges have prevented a wider adoption of adaptive designs
in clinical trials despite their ethical advantages \parencite{robertson_point_2023}. In particular, because Box's linear process of ``conjecture-design-experiment-analysis'' no longer applies to adaptive experiments, most standard statistical methods exhibit selection bias \parencite{rosenberger_randomization_2015}. 
For certain parametric models, the distribution of a test statistic in
an adaptive study can be derived \parencite{spencer_adaptive_2016,
  lin_inference_2021,frieri_design_2022}, or one can conduct hypothesis
tests for the different stages of the trial separately and combine the
resulting p-values \parencite{pocock_interim_1982,
  muller_adaptive_2001, stallard_adaptive_2023}.
However, these approaches are limited to just a few trial designs and
heavily rely on modelling assumptions that are often difficult to
justify. In the last decade, applications in online experiments
of web services sparked renewed interest in adaptively
collected data. A recent line of research uses adaptive weighting
techniques to recover asymptotic normality
\parencite{deshpande_accurate_2018, zhang_inference_2020,
  hadad_confidence_2021, zhang_statistical_2021}; \revc{moreover, \textcite{hirano_asymptotic_2025} develop asymptotic representations for adaptive batched experiments.} In another novel
approach, \textcite{howard_time-uniform_2021} use martingale
concentration inequalities to construct any-time valid confidence
intervals for the average treatment effect. These new methods usually
require weaker assumptions but are often technically involved and only
apply to simple response-adaptive designs.

Selective inference seeks valid inference on potential findings that
are discovered after viewing the data
\parencite{benjaminiSelectiveInferenceSilent2020}. Several notions of
``validity'' (or selective risks) have been considered. The notion
with the longest history is familywise error, also known as
simultaneous inference, which dates back to the well known
methods of \textcite{tukeyComparingIndividualMeans1949} and
\textcite{scheffeAnalysisVariance1959}. Although simultaneous
inference is often associated with multiple hypothesis testing, it can
also be applied to post-model selection problems
\parencite{berkValidPostselectionInference2013}. A more lenient notion
of risk
is the average error rate over the selected; examples include the
false discovery rate
\parencite{benjaminiControllingFalseDiscovery1995} and the false
coverage-statement rate
\parencite{benjaminiFalseDiscoveryRate2005a}. A third notion, which is
the one used in our proposed selective randomization test, is the
conditional post-selection error
\parencite{fithian_optimal_2017}. Interest in this approach has surged
after successful applications to the model selection problem
\parencite{lee_exact_2016}; see
\textcite{kuchibhotla_post-selection_2022} for a recent review.

\begin{figure}[t]
  \centering
  \begin{subfigure}[t]{0.4\textwidth} \centering
    \begin{tikzcd}
      \text{conjecture} \arrow[r] & \text{design} \arrow[d] \\
      & \text{experiment} \arrow[d] \\
      & \text{final analysis}
    \end{tikzcd}
    \caption{Box's linear process.}
    \label{fig:box}
  \end{subfigure}
  \quad
  \begin{subfigure}[t]{0.4\textwidth} \centering
    \begin{tikzcd}
      \text{conjecture} \arrow[r] & \text{design} \arrow[d] \\
      \text{prelim analysis} \arrow[ur] & \text{experiment} \arrow[l] \arrow[d] \\
      & \text{final analysis}
    \end{tikzcd}
    \caption{Response-adaptive randomization.}
    \label{fig:rar}
  \end{subfigure} \\[2em]

  \begin{subfigure}[t]{0.4\textwidth} \centering
    \begin{tikzcd}
      \text{conjecture} \arrow[r] & \text{design} \arrow[d] \\
      \text{prelim analysis} \arrow[u] & \text{experiment} \arrow[l] \arrow[d] \\
      & \text{final analysis}
    \end{tikzcd}
    \caption{Enrichment design.}
    \label{fig:enrichment}
  \end{subfigure}
  \quad
  \begin{subfigure}[t]{0.4\textwidth} \centering
    \begin{tikzcd}
      \text{conjecture} \arrow[ddr] & \text{design} \arrow[d] \\
      \text{prelim analysis} \arrow[u] & \text{experiment} \arrow[l] \\
      & \text{final analysis}
    \end{tikzcd}
    \caption{Post-selection inference.}
    \label{fig:posi}
  \end{subfigure} \\[2em]

  \begin{subfigure}[t]{0.4\textwidth} \centering
    \begin{tikzcd}
      \text{conjecture} \arrow[r] \arrow[ddr] & \text{design} \arrow[d] \\
      \text{prelim analysis} \arrow[u] \arrow[ur] & \text{experiment} \arrow[l] \arrow[d]\\
      & \text{final analysis}
    \end{tikzcd}
    \caption{Our problem.}
    \label{fig:ours}
  \end{subfigure}

  \caption{A schematic illustration of different
    modes of scientific research. Adaptive experimentation
    (\Cref{fig:rar,fig:enrichment}) and post-selection
    inference (\Cref{fig:posi}) break away from traditional
    statistical inference by inserting different ``mini cycles''
    within each cycle of conjecture-design-experiment-analysis
    described by \textcite{box1957}.}
  \label{fig:cycles}
\end{figure}

\Cref{fig:cycles} illustrates how adaptive experiments and selective inference refine Box's linear process of conjecture-design-experiment-analysis (\Cref{fig:box}). In adaptive experiments, preliminary data from the
experiment are used to modify the experimental design (\Cref{fig:rar}) and conjecture (\Cref{fig:enrichment}). In post-selection inference,
preliminary analyses (e.g.\ a model selection procedure) are used to
choose the conjecture (\Cref{fig:posi}). As mentioned previously, the
selective randomization test proposed here applies conditional post-selection inference to
adaptive experiments (\Cref{fig:ours}).

Two points merit a discussion at this point. First, among the
different approaches to selective inference, we believe conditional
inference is the most appropriate for adaptive clinical trials. This
is because controlling the selective type-I error as in
\eqref{eq:selective-error} allows us to \emph{treat the selected
  hypothesis as given} when interpreting the results. This is crucial
in objectively evaluating the evidence in a randomized experiment from
a regulatory perspective. The familywise error rate is unnecessarily
strict as we are often just interested in testing one hypothesis, and
the more lenient notions such as the false discovery rate are more
suitable for discovering promising hypotheses but not for confirmatory
studies.

Second, compared to existing methods for inference after model
selection, a distinguishing feature of the selective randomization
test is that it is model-free. In particular, it does not require any
parametric model on the outcome, not even that the data points are
independent or identically distributed. Originally suggested by
\textcite{fisher_design_1935} and elaborated by
\textcite{pitman_significance_1937}, randomization inference is
exactly based on the randomness introduced by the experimenter; in
Fisher's words, the ``physical act of randomization'' justifies the
statistical inference. Although model-based inference is predominant
in the analysis of randomized experiments, interest in randomization
inference, especially for complex designs, has rekindled in recent years
\parencite{ji_randomization_2017,
  rosenberger_randomization_2019,wang_randomization_2020,
  offer-westort_adaptive_2021, bugni_inference_2018}. \revc{In
  particular, randomization tests can be easily tailored to
  non-standard settings by conditioning on a suitable statistic 
  \parencite{morgan_rerandomization_2012,zheng_multi-center_2008,basse2019randomization};
  see \textcite{zhang_2023_randomization_test} for a recent review on
  conditional randomization tests.} Some previous works
examined the use of randomization inference in adaptive
experiments \parencite{edwards_model_1999,simon2011using}, but the
selective randomization test proposed here can be applied to a much
wider class of adaptive designs.

\subsection{Organization of the Article}
Section~\ref{sec:notation} introduces our notation and causal
graphical framework for a general adaptive study with multiple
stages. Section~\ref{sec:selective-inference} formally defines the
selective randomization p-value, \revc{discusses under which
  conditions it is computable} and explains how it can be used to
estimate and construct confidence intervals for a homogeneous
treatment effect. In Sections~\ref{sec:simulation-study} and
\ref{sec:real.data}, we illustrate the validity of our method using
simulations and real-world data. We conclude with a discussion on future
research questions in
Section~\ref{sec:discussion}. The Supplementary Material contain
technical proofs, Monte Carlo methods for computing the selective
randomization p-value, and details about the simulation study.

\section{Adaptive Multi-stage Experiments}\label{sec:notation}


\revc{This section provides a general description of an adaptive experiment with multiple stages. First, we introduce our notation and a graphical model which provides an intuitive depiction of the experiment. Then, we describe the minimal assumptions that allow us to conduct selective randomization inference. Fulfilling these conditions is thus an important consideration for the design of an adaptive experiment. Lastly, we provide simplified versions of the assumptions (marked by $*$) and explain how they can be understood and ensured in practice. Throughout, we use the two-stage enrichment trial in Example~\ref{ex:enrichment} to illustrate our notation and assumptions.}

\subsection{Experimental Set-up and Notation}

Consider an 
adaptive study with $K$ stages and $L+1$ different treatments (e.g.\
drugs, dosages, policies) that are enumerated from $0$ to $L \geq
1$. Denote $[K] := \{1,\ldots,K\},\, [0]:=\emptyset$, and we will use
$k\in [K]$ to denote an arbitrary stage and $l \in \{0,\ldots,L\}$ an
arbitrary treatment. At every stage, both the recruitment of
participants and the treatment assignment mechanism may depend on data
from previous stages. \revc{In order to visually describe the relationships of the involved variables both within and between stages, we employ directed acyclic graphs (DAGs). For instance, the graph in Figure~\ref{fig:dag} depicts an adaptive experiment with $K=2$ stages, like Example~\ref{ex:enrichment}. Directed edges like $X_{R_1} \rightarrow Z_1$ can be interpreted as ``$X_{R_1}$ influences $Z_1$'' or ``$Z_1$ is a function of $X_{R_1}$''. In essence, Figure~\ref{fig:dag} is a mathematical formalization of the schematic in Figure~\ref{fig:trial-pictogram}. We return to these graphical models at the end of this section, but first, we describe the setup more closely.}

\begin{figure}[t] \centering
  \begin{tikzpicture}
    \node at (0,0.2) (yr1) {$Y_{R_1}$};
    \node at (-3, 2) (r1) {$R_1$};
    \node at (-1.25, 1) (xr1) {$X_{R_1}$};
    \node at (-3,-1) (w1) {$Z_1$};
    \node at (2,0.2) (s1) {$S_1$};
    \node at (0,2) (yr1s) {$Y_{R_1}(\cdot)$};

    \node at (7,0.2) (yr2) {$Y_{R_2}$};
    \node at (4,2) (r2) {$R_2$};
    \node at (5.75,1) (xr2) {$X_{R_2}$};
    \node at (4,-1) (w2) {$Z_2$};
    \node at (9,0.2) (s2) {$S_2$};
    \node at (7,2) (yr2s) {$Y_{R_2}(\cdot)$};

    \node at (2.5,4.5) (ys) {$Y(\cdot)$};
    \node at (1 ,4.5) (x) {$X$};

    \draw[-latex, thick] (x) -- (ys);

    \draw[-latex, thick] (xr1) -- (w1);
    \draw[-latex, thick] (w1) -- (yr1);
    \draw[-latex, thick] (yr1) -- (s1);
    \draw[-latex, thick] (xr1) -- (s1);
    \draw[-latex, thick] (w1) -- (s1);
    \draw[-latex, thick] (yr1s) -- (yr1);
    \draw[-latex, thick] (ys) -- (yr1s);
    \draw[-latex, thick] (r1) -- (w1);
    \draw[-latex, thick] (r1) -- (xr1);
    \draw[-latex, thick] (r1) -- (yr1s);

    \draw[-latex, thick] (xr2) -- (w2);
    \draw[-latex, thick] (w2) -- (yr2);
    \draw[-latex, thick] (yr2) -- (s2);
    \draw[-latex, thick] (xr2) -- (s2);
    \draw[-latex, thick] (w2) -- (s2);
    \draw[-latex, thick] (yr2s) -- (yr2);
    \draw[-latex, thick] (ys) -- (yr2s);
    \draw[-latex, thick] (r2) -- (yr2s);
    \draw[-latex, thick] (r2) -- (w2);
    \draw[-latex, thick] (r2) -- (xr2);

    \draw[-latex, thick] (s1) -- (r2);
    \draw[-latex, thick] (s1) -- (w2);

    \draw[-latex, thick, color=black] (s1) .. controls (2.2,-2.8) and (8.5,-3.2) .. (s2);

    \draw[-latex, thick, color=black] (r1) .. controls (-2,3.25) and (2.5,3.25) .. (r2);

    \draw[-latex, thick, color=black] (x) .. controls (-0.5,4.5) and (-2.5,3.5) .. (r1);

    \draw[-latex, thick, color=black] (x) .. controls (1.3,4.3) and (3,3.5) .. (r2);


    \draw[-latex, thick, color=black] (x) .. controls (-0.75,3.5) and (-1,2.2) .. (xr1);

    \draw[-latex, thick, color=black] (x) .. controls (3,4) and (5,2.5) .. (xr2);

    \begin{pgfonlayer}{background}
      \node [fill=blue!20,fit=(w1) (yr1s) (xr1) (yr1), text=blue, align=right,text height=115pt] {\phantom{blabla}Stage 1};
      \node [fill=blue!20,fit=(w2) (yr2s) (xr2) (yr2), text=blue, align=right, text height=115pt] {\phantom{blabla}Stage 2};
      \node [fill=red!20,fit=(s1)] {};
      \node [fill=red!20,fit=(s2)] {};
      \node [fill=black!20,fit=(x) (ys)] {};
    \end{pgfonlayer}

  \end{tikzpicture}
  \caption{Directed acyclic graph (DAG) of a two-stage adaptive trial, where Assumptions~\eqref{eq:A2-star} and~\eqref{eq:A3-star} hold. In addition, $S_2$ may directly depend on the observed data from the first stage: $R_1, X_{R_1}, Z_1$ and $Y_{R_1}$. For simplicity, these relationships are not depicted.}
  \label{fig:dag}
\end{figure}

We assume participants are recruited from a pool of $n$ units, \revc{e.g.\ all people registered with a health service or all patients in a hospital during a certain time period}. For each candidate $i \in [n]$, we assume some covariate~$X_i$ would become available if the candidate is recruited; let us denote $X := (X_i)_{i\in [n]}$. Following the Neyman-Rubin causal model, the causal effect of a treatment~$l$ on a unit~$i$ is described using the potential outcome~$Y_i(l)$. Let $Y_i(\cdot) := (Y_i(l))_{l \in \{0,\dots,L\}}$ denote all the potential outcomes of participant~$i$, and let $Y(\cdot) := (Y_i(\cdot))_{i \in [n]}$. Here, we implicitly make the familiar assumption of ``no interference'', meaning the potential outcomes of one unit do not depend on the treatments received by the other units. Finally, for a subset $I\subseteq [n]$ of units, we use the notations $X_I := (X_i)_{i \in I}$ and $Y_I(\cdot) := (Y_i(\cdot))_{i \in I}$ . We refer the reader to \textcite{imbens2015causal} for a comprehensive introduction to the Neyman-Rubin potential outcomes model. \revc{In Example~\ref{ex:enrichment}, there are two stages and we study two treatments (the new therapy and a placebo); hence, $K=2$ and $L=1$. For every unit $i$ participating in the study, the corresponding covariates $X_i$ include the genetic risk score and the potential outcomes $Y_i(0)$ and $Y_i(1)$ describe the patient's outcome, such as blood pressure or cholesterol level, under the placebo and the new therapy, respectively.}

The first distinguishing feature of our setup is that recruitment is
modelled as a random event. In the clinical trials and causal
inference literature, it is usually assumed implicitly that
the recruited units are fixed. 
However, in adaptive experiments such as enrichment trials, recruitment decisions can depend on previously observed data \revd{as well as the covariates $X$ of the units in the recruitment pool}. To incorporate this, we model the indices of units participating in each stage as the random sets $R_1, \ldots, R_K \subseteq [n]$.
To simplify the problem, we do not consider longitudinal studies with patient follow-up over multiple time points in this work and assume that each unit can be included into the trial at most once, so $R_1,\ldots, R_K$ are disjoint.
This distinguishes our set-up from sequential multiple assignment randomized trials (SMART), where the same participants pass through multiple stages and receive interventions/treatments that can be adapted to their own past outcomes \parencite{almirall_introduction_2014}.
The covariates and potential outcomes of participants at stage $k$ are denoted by $X_{R_k}$ and $Y_{R_k}(\cdot)$, respectively. Denote ${R= \cup_k R_k}$, $R_{[k]} = \cup_{k'=1}^k R_{k'}$ and $R^C = [n]\setminus R$. 

\subsection{Assumptions}

At stage $k$, a total of $|R_k|$ units are recruited; let their administered treatments be denoted as $Z_k \in \calZ_k \subseteq \{0,\ldots,L\}^{\lvert R_k\rvert}$. 
The treatment assignment determines which of the potential outcomes is actually realized: the observed values of the outcome variable are denoted by $Y_{R_k}$. We assume consistency between the potential outcomes and factual outcomes, 
that is, $Y_{i} = Y_{i}(Z_{k,i})$ for all $i \in R_k$. 
In the following, we use the short-hand notations $Z = (Z_k)_{k \in [K]}$, $\calZ = \bigtimes_{k=1}^K\, \calZ_k$ and $Y_R = (Y_{R_k})_{k \in [K]}$. 
In adaptive experiments, $Z_k$ is randomized by the experimenter and may depend on the covariates of the recruited participants $X_{R_k}$ and \emph{observed} data from previous stages. 
Accordingly, we assume that the \revd{sequential} treatment assignment distribution 
\begin{equation}\label{eq:A1}
    q(z\mid r, x_r, y_r) := \prod_{k=1}^K \Prob(Z_k = z_k \mid R_{[k]} = r_{[k]}, X_{R_{[k]}} = x_{r_{[k]}}, Y_{R_{[k-1]}} = y_{r_{[k-1]}}, Z_{[k-1]} = z_{[k-1]})\quad\text{is known}. \tag{A1}
\end{equation}
\revc{Going forward, we refer to the $k$-th factors on the right hand
  side as $q_k$.}

\revc{Moreover, we need a multi-stage generalization of the typical assumption that the treatment assignment is conditionally independent of the potential outcomes; in the usual setting with one stage, this condition is expressed as $Z_1 \indep Y_{R_1}(\cdot)\mid R_1, X_{R_1}$. To this end, } we require that for any stage $k$ the treatment assignment $Z_k$ depends on the potential outcomes of the current and previous stages only through already observed data. We formalize this \revc{assumption} as the following conditional independence relationship:
\begin{equation}\label{eq:A2}
    Z_k \indep Y_{R_{[k]}}(\cdot) \mid R_{[k]}, X_{R_{[k]}}, Y_{R_{[k-1]}}, Z_{[k-1]},\qquad \forall\, k \in [K]. \tag{A2}
\end{equation}

\revc{Note that the treatment $Z_k$ is allowed to depend on \emph{observed} outcomes from previous stages as we condition on $Y_{R_{[k-1]}}$.}
Assumption~\eqref{eq:A2} is akin to the different notions of
sequential exchangeability as described in \textcite[chap.\
19.5]{hernan_causal_2024}. Unlike their setting, we do not consider
time-varying treatments and we explicitly model recruitment as a
random variable.

\revc{In an adaptive trial, Assumptions~\eqref{eq:A1}
  and~\eqref{eq:A2} can be guaranteed because the experimenter
  randomizes the treatment assignments. While~\eqref{eq:A1} is then
  obviously satisfied, Assumption~\eqref{eq:A2} is true whenever $Z_k$
  is a function of already observed data and an exogenous source of
  randomness $E_k$, for instance a computer-generated random number;
  that is, we can write $Z_k = f_k(R_{[k]}, X_{R_{[k]}},
  Y_{R_{[k-1]}}, Z_{[k-1]}, E_k)$. 
  } 
  Consider, for instance, a simple Bernoulli trial with two arms (i.e.\ $L=1$): at
each stage $k \in [K]$, every recruited participant is independently
assigned to treatment or control with a pre-specified probability $p
\in (0,1)$. Hence, the \revd{sequential} \revc{treatment assignment distribution} is
given by
\revc{
\begin{equation*}
    q(z\mid r,x_r,y_r) = q(z \mid r) = \prod_{k=1}^K\, \prod_{i \in r_k} p^{z_{k,i}}\,(1-p)^{1-z_{k,i}}.
\end{equation*}
}
Another simple example is the completely randomized design (or
sampling without replacement) that enforces a fixed number of
treatment and control units. These designs can be generalized by
allowing the probability of treatment to depend on
observed covariates; examples include randomized complete block design,
biased coin design \parencite{efron_forcing_1971}, and rerandomization
\parencite{morgan_rerandomization_2012}. See \textcite{ye_toward_2023}
and the references therein for more examples of covariate-adaptive
designs.

The second characterizing feature of our setting is the selection statistics $S_1,\dots,S_K$, which describe the adaptive nature of the experiment. After administering the treatment and recording the response in stage $k$, the experimenter is allowed to give a preliminary analysis of the data from stages~$1$ through $k$ and summarize their findings in a statistic~$S_k$, which then determines the design of the next stage of the experiment. In other words, the recruitment and treatment assignment in stage $k+1$ are allowed to depend on $S_k$; in \Cref{fig:dag}, this corresponds to the directed edges $S_1 \rightarrow R_2$ and $S_1 \rightarrow Z_2$. 
\revd{Hence, the recruitment in the $k$-th stage is a function of $R_{[k-1]}$, $X$, $S_{k-1}$ and some exogenous source of randomness $\tilde{E}_k$, i.e.\ we can write $R_k = \tilde{f}_k(R_{[k-1]}, X, S_{k-1}, \tilde{E}_k)$.}

Note that $S_k$ is a function of the previously observed data $R_{[k]}, X_{R_{[k]}}, Y_{R_{[k]}}$ and $Z_{[k]}$. After completing the trial, the researchers choose a null hypothesis that may depend on the observed data of all stages of the trial. Hence, we model it as the final selection statistic $S_K$. We set $S:=(S_1,\ldots,S_K) \in \calS = \calS_{-K}\times \calS_K$. In general, $S$ is a deterministic function of \revc{$R$, $X_R$, $Y_R(Z)$} and $Z$ but we often shorten the notation as $S := S(Z) := S(Z,R,X_R,Y_R)$; moreover, we use the convention $S_0 := \emptyset$.
\revc{In our running example~\ref{ex:enrichment}, we may construct the selection statistic $S_1$ as follows. We estimate the difference $\Delta$ of the treatment effects between the two groups on data from the first stage; then, we discretize it to three levels that correspond to recruiting participants with low genetic risk, high genetic risk or both groups in the second stage. Consequently, $S_1$ is a categorical statistic that can assume three values. The mathematical expressions can be found in Section~\ref{sec:simulation-study}.}

\revc{To ensure type-I error control,} it is important
that the selection statistic $S_{k-1}$ captures \emph{all the
  information} in previously observed data that impacts the design of
the stage $k$. This can be formalized via the following conditional independence statement:
\begin{equation}\label{eq:A3}
    R_k,X_{R_k},Y_{R_k}(\cdot) \indep Z_{[k-1]} \mid R_{[k-1]},X_{R_{[k-1]}}, Y_{R_{[k-1]}}(\cdot), S_{k-1},\qquad \forall\,k \in [K].\tag{A3}
\end{equation}
\revc{Assumption~\eqref{eq:A3} stipulates that the recruitment, covariates and potential outcomes in stage $k$ are conditionally independent of the treatments assigned in earlier stages $Z_{[k-1]}$ given earlier recruitment decisions, covariates, potential outcomes and, importantly, the statistic $S_{k-1}$. This condition can be ensured by only revealing $S_{k-1}$ and $(R_{[k-1]},X_{R_{[k-1]}}, Y_{R_{[k-1]}})$ to the experimenter recruiting participants for the $k$-th stage.}


\revc{Assumptions~\eqref{eq:A1},~\eqref{eq:A2} and~\eqref{eq:A3} are the minimal conditions that we require for the selective randomization test. They allow for highly adaptive recruitment decisions and treatment assignment distributions. In practice, many designs fulfill stronger conditional independence relationships that are easier to verify. For instance, \eqref{eq:A2} can be strengthened as
\begin{equation}
   Z_k \indep Z_{[k-1]}, R_{[k-1]}, X_{R_{[k-1]}}, Y_{R_{[k]}}(\cdot) \mid R_k, X_{R_k},S_{k-1},\qquad \forall\,k \in [K]. \tag{A2*} \label{eq:A2-star}
\end{equation}
This assumption states that the treatment $Z_k$ depends on (observed and counterfactual) data in previous stages only through the recruitment and covariates in stage $k$ and the selection statistic $S_{k-1}$. One notable feature of designs satisfying \eqref{eq:A2-star} is that the selective randomization p-value can be computed in parallel; see \Cref{prop:computability} in the next section.

Similarly, we can strengthen Assumption~\eqref{eq:A3} by requiring that the recruitment in stage $k$ depends on data from previous stages only through $S_{k-1}$:
\begin{subequations}
    \makeatletter
        \def\@currentlabel{A3*}
        \makeatother
        \label{eq:A3-star}
        \renewcommand{\theequation}{A3*-\arabic{equation}}
    \begin{alignat}{2}
        R_k &\indep Z_{[k-1]}, R_{[k-1]}, X_{R_{[k-1]}}, Y_{R_{[k-1]}}(\cdot) \mid S_{k-1},\qquad & &\forall\,k \in [K], \label{eq:A3-star-1}\\
        X_{R_k},Y_{R_k}(\cdot) &\indep Z_{[k-1]} \mid R_{[k]}, X_{R_{[k-1]}}, Y_{R_{[k-1]}}(\cdot), S_{k-1},\qquad & &\forall\,k \in [K]. \label{eq:A3-star-2}
    \end{alignat}
\end{subequations}
\addtocounter{equation}{-1}
Note that \eqref{eq:A3-star-1} and~\eqref{eq:A3-star-2} imply \eqref{eq:A3} by using the weak union and contraction axioms of conditional independence \parencite{dawid1979conditional,pearlProbabilisticReasoningIntelligent1988}. From here on, we refer to Assumptions~\eqref{eq:A3-star-1} and~\eqref{eq:A3-star-2} collectively as~\eqref{eq:A3-star}.}

\revc{

A more intuitive way to understand and verify Assumptions~\eqref{eq:A2} and~\eqref{eq:A3} is through a causal diagram. During the planning phase of a study, the experimenters specify all involved variables (most importantly the selection statistics) and draw a directed acyclic graph (DAG) that depicts \emph{all} the dependencies between them according to the proposed design; see Figure~\ref{fig:dag} for an example. Recall that directed edges like $X_{R_1} \rightarrow Z_1$ are interpreted as ``$X_{R_1}$ influences $Z_1$'' or ``$Z_1$ is a function of $X_{R_1}$''. Once the DAG is specified, we can use the d-separation criterion~\parencite{pearlProbabilisticReasoningIntelligent1988} to read off conditional independencies from the graph. For instance, the two-stage trial in Figure~\ref{fig:dag} fulfills Assumption~\eqref{eq:A3-star-1} as all paths from $(Z_1, R_1, X_{R_1}, Y_{R_1}(\cdot))$ to $R_2$ are blocked by $S_1$. Furthermore, \eqref{eq:A3-star-2} holds as all paths from $Z_1$ to $(X_{R_2}, Y_{R_2}(\cdot))$ are blocked by $(R_1, R_2, X_{R_1}, Y_{R_1}(\cdot), S_1)$.

}

\section{Selective Randomization Inference}\label{sec:selective-inference}

This section applies ideas from the literature on conditional post-selection inference to randomization tests in adaptive experiments. We will formally define the selective randomization p-value, \revc{discuss under which conditions it is computable and address estimation and confidence intervals for a homogeneous treatment effect.}

\subsection{Selective Randomization Distribution}\label{sec:sel-rand-dist}
First suggested in \textcite{fisher_design_1935}, randomization
inference leverages the fact that the conditional distribution of the
treatment assignments~$Z$ is specified by the experimenter. In order
to avoid imposing any modelling assumptions on the potential outcomes
$Y(\cdot)$ and covariates~$X$ that may be unrealistic, randomization
inference conditions on these quantities. In most settings, the
recruited units are regarded as fixed and given. 
In adaptive studies, however, the recruitment
of participants~$R$ may depend on the treatment assignments~$Z$. 

Note that there are two sources of randomness in the experiment: the recruitment $R_k$, which is modelled as a random event, and the randomized treatment assignment given $R_k$, which is the basis of randomization inference. To restrict our analysis to the recruited units, we additionally condition on $R$; as a consequence, randomization inference achieves high internal validity but requires additional assumptions for external validity. 
This yields the randomization distribution
\revc{
    ${\Prob(Z=\cdot \mid R, X_R, Y_R(\cdot)).}$
}

\revd{Randomization tests based on this distribution}
generally do not control the false positive (aka type-I error) rate in adaptive studies, because some information in earlier stages of the study is already used to choose the null hypothesis and design in the later stages. To avoid ``double dipping'', we can remove any such information by additionally conditioning on the selection statistics $S = (S_1,\dots,S_K)$ which may depend on $Z$ \parencite{fithian_optimal_2017}.
The resulting \emph{selective randomization distribution} is then given by
\revc{
    $\Prob(Z=\cdot \mid R, X_R, Y_R(\cdot), S(Z)).$
}

Next, we will use this distribution to define the selective
randomization p-value and show that it controls the type-I error rate
for the \emph{selected} hypothesis.

\subsection{Hypothesis Testing}\label{sec:hypothesis-testing}

Causal hypotheses about the efficacy of the treatments can be formulated as contrasts of (a subset of) the corresponding potential outcomes.
Since one of the potential outcomes of each recruited unit is observed
under the consistency assumption, that is, $Y_R(Z) = Y_R$, the null
hypothesis allows us to impute (a subset of) the remaining~$Y_R(\cdot)$ from~$Y_R$. A common hypothesis in randomization inference is Fisher's sharp null which states that there is no treatment effect whatsoever:
\begin{equation*}
    \mathrm{H}_0:\quad Y_i(l)=Y_i(l'),\quad \text{for all } l,l'\in \{0,\ldots,L\},\, i \in R.
\end{equation*}
Under this hypothesis, all potential outcomes can be imputed from the
observed outcomes, so \revc{$Y_R(\cdot)$ and thus ${\Prob(Z=\cdot \mid R, X_R, Y_R(\cdot), S(Z))}$ is known. To test Fisher's sharp null, we may now choose a test statistic $T := T(Z, R, X_R,Y_R(\cdot))$} and compare its observed value to its selective randomization distribution under the null.

Yet, in many adaptive studies the null hypothesis of interest \revc{does not apply to all potential outcomes. Therefore, we can only impute a subset of} $Y_R(\cdot)$ from the observed $Y_R$. \textcite{zhang_2023_randomization_test} call such a null hypothesis \emph{partially sharp}. \revc{Let us consider an example:
Suppose there are $L+1 = 3$ treatments but we are only interested in comparing the first two by testing the null hypothesis
\begin{equation*}
    \mathrm{H}_0:\quad Y_i(0)=Y_i(1),\quad \text{for all } i \in R.
\end{equation*}
In this case, data from units who do not receive the first two treatments are not useful and we cannot impute any potential outcomes $Y_i(2)$ from this null hypothesis. Therefore, we need to condition on the set of units that receive the first two treatments, i.e.\ we condition on $G(Z) = (\one\{Z_{k,i} \in \{0,1\}\})_{k,i}$. Abstractly speaking, when a partially sharp null hypothesis is specified, we usually need to restrict the support of the (selective) randomization distribution by conditioning on an additional statistic $G := G(Z):=G(Z,R,X_R)$. This also applies to the naive and the data splitting approach and is necessary to render the respective p-values computable. For more details, see Section~\ref{sec:computability} and Supplementary Material~\ref{supp:computability}.
}

In summary, conditioning on $S$ accounts for the selection of design and null hypothesis, and conditioning on~$G$ accommodates partially sharp null hypotheses. \revc{To test such a null hypothesis, we choose a test statistic~$T$ and compare its observed value to its selective randomization distribution
\begin{equation*}
    \Prob\big(\,T(Z,R,X_R,Y_R(\cdot))=\cdot\,\big\vert\, R,X_R,Y_R(\cdot), S(Z), G(Z)\big).
\end{equation*}
We define the selective randomization p-value in this section and elaborate on conditions concerning $T$, $S$, $G$ and~$q$ that render it computable in the next section.
} 


\begin{definition}\label{def:rand-p-value}
    \revd{In the setting above, let $Z^*$ be independently drawn from the sequential treatment assignment distribution, that is for all $k \in [K]$:
    \begin{align}
        Z^*_k &\overset{D}{=} Z_k \mid R_{[k]}, X_{R_{[k]}},
        Y^*_{R_{[k-1]}} = Y_{R_{[k-1]}}, Z^*_{[k-1]} = Z_{[k-1]},
        \label{eq:zstar-id}\\
        \begin{split}
            Z_k^* &\indep Z_{[k]}, Y_{R_{[k]}}(\cdot) \mid R_{[k]}, X_{R_{[k]}}, Y^*_{R_{[k-1]}}, Z^*_{[k-1]}, \\
            Z_k   &\indep Z^*_{[k]}, Y_{R_{[k]}}(\cdot) \mid R_{[k]}, X_{R_{[k]}}, Y_{R_{[k-1]}}, Z_{[k-1]},
        \end{split}
        \label{eq:zstar-indep}\\
        R_k, X_{R_k}, Y_{R_k}(\cdot) &\indep Z^*_{[k-1]} \mid R_{[k-1]}, X_{R_{[k-1]}}, Y_{R_{[k-1]}}(\cdot), Z_{[k-1]}, S^*_{k-1}=S_{k-1},
        \label{eq:zstar-sel}
    \end{align}
    where $Y^*_{R_{k}} = Y_{R_k}(Z_k^*)$ and $S^*_{k-1}=S_{k-1}(Z^*_{[k-1]})$.}
    Then, the \emph{selective randomization p-value} is given by
    \begin{equation}
        P_{\text{sel}}:=\Prob\big(T(Z^*\!,\,R,X_R,Y_R(\cdot)) \leq T(Z,\,R,X_R,Y_R(\cdot)) \,\big\vert\, R,X_R,Y_R(\cdot),\,\, Z,\,\, S(Z^*) = S(Z), G(Z^*) = G(Z)\big),\label{eq:rand-p-value}
    \end{equation}
    which is a function of $R,X_R,Y_R(\cdot)$ and $Z$ in general.
\end{definition}

\revd{
\begin{remark} The conditions on $Z^*$ in the definition above are seemingly complex but naturally generalize the non-adaptive case where $K=1$. There, $Z_1$ and $Z^*_1$ are assumed to have the same distribution and satisfy the conditional independence $Z^* \indep Z \indep Y_{R_1}(\cdot) \mid R_1, X_{R_1}$. This is extended by~\eqref{eq:zstar-id} and~\eqref{eq:zstar-indep}, respectively. Lastly, we also require that conditioning on $S$ captures all the selective information under the treatment $Z^*$, as well. This is formalized in~\eqref{eq:zstar-sel} and akin to Assumption~\eqref{eq:A3}.
\end{remark}
}

The selective randomization p-value is a specific instance
of the conditional randomization p-value defined in
\textcite{zhang_2023_randomization_test}. Therefore, the selective
type-I error control follows immediately from Theorem~1 in
\textcite{zhang_2023_randomization_test}.

\begin{proposition}\label{thm:type-i-error}
    Let $\alpha \in [0,1]$. The selective randomization p-value as given in Definition \ref{def:rand-p-value} stochastically dominates the uniform distribution on $[0,1]$:
    \begin{equation*}
        \Prob(\,P_{\textnormal{sel}} \leq \alpha \,\vert\, R,X_R,Y_R(\cdot), S(Z)=s, G(Z)=g) \leq \alpha,\quad \forall\, s\in \calS, g\in \calG.
    \end{equation*}
    Consequently, a test which rejects the selected null hypothesis when $P_{\textnormal{sel}} \leq \alpha$ controls the selective type-I error, i.e.
    \begin{equation*}
        \Prob(\,P_{\textnormal{sel}} \leq \alpha \mid R,X_R,Y_R(\cdot), S(Z)=s) \leq \alpha,\quad \forall\, s\in \calS.
    \end{equation*}
\end{proposition}

\begin{remark}
    \textcite{simon2011using} consider adaptive experiments 
    in which the recruitment scheme and null hypothesis are fixed, but
    the treatment assignment distribution may depend on data from
    previous stages. They implicitly assume ignorability, cf.\
    Assumption~\eqref{eq:A2}, and prove that, under a condition akin
    to Assumption~\eqref{eq:A3}, the usual/naive randomization p-value
    is valid. Their result is a special case of
    Proposition~\ref{thm:type-i-error}: the indices of the recruited
    people are deterministic $R=[n]$; moreover, their assumptions
    imply that $S$ can only be function of $R,X_R$ and $Y_R(\cdot)$ but not $Z$, so $S$ can
    be dropped from the conditioning event as $(R,X_R,Y_R(\cdot))$ is already
    conditioned on.
\end{remark}



\revc{
\subsection{Computability}\label{sec:computability}


As shown in Proposition~\ref{thm:type-i-error}, the p-value $P_{\text{sel}}$ controls the selective type-I error for every choice of test statistic~$T$. Yet, it depends on unobserved potential outcomes; therefore, it is not immediately clear if it can actually be calculated from the data.

If the null hypothesis of interest is sharp, we can indeed infer/impute \emph{all} unobserved potential outcomes from the observed outcomes $Y_R$; hence, $Y_R(\cdot)$ is known and we can compute the selective randomization p-value. If the null hypothesis is partially sharp, however, only a subset of unobserved potential outcomes can be imputed from~$Y_R$. For instance, if we only test for ``no effect whatsoever'' in the group with high genetic risk score, the null hypothesis does not allow us to impute unobserved potential outcomes in the group with low genetic risk score. To ensure that the test statistic~$T$ can still be evaluated in such a situation, we restrict the space of alternative treatments by additionally conditioning on $G(Z)=G(Z,R,X_R)$ such that for any $z^*\in \calZ$ with $G(z^*) = G(Z)$ the value $T(z^*)$ only depends on the imputable potential outcomes. In the example above, we could choose $G(Z)$ as the treatments assigned to the units with low genetic risk score. In adaptive experiments, not only the test statistic but also the selection statistics and the treatment assignment probabilities can be outcome-dependent. Therefore, we need to additionally ensure that conditioning on $G(Z)$ renders not only $T$ but also $S$ and $q$ imputable. For a more detailed and rigorous discussion of imputability, we refer the reader to Supplementary Material~\ref{supp:computability}.

We can now state our main result on the computability of the selective randomization p-value.
}

\begin{proposition}\label{prop:computability}
Under Assumptions~\eqref{eq:A2} and~\eqref{eq:A3}, \revd{the selective randomization distribution is given by}
\begin{equation}\label{eq:factorization-new}
    \Prob(Z = z\mid R = r, X_R = x_r, Y_{R}(\cdot)=y_{r}(\cdot), S(Z)=s, G(Z)=g) \,=\,  \frac{\one{\{G(z)=g, S(z)=s\}}\cdot q(z \mid r,x_r,y_r)}{\sum_{z' \in \calZ}\, \one{\{G(z')=g, S(z')=s\}}\cdot q(z' \mid r,x_r,y_r)}.
\end{equation}
If Assumption~\eqref{eq:A1} holds and $T$, $q$ and $S$ are imputable under the null hypothesis with respect to $G$, then the selective randomization p-value is computable and is given by
\begin{equation}\label{eq:p-sel-formula}
    P_\textnormal{sel} =
    \frac{\sum\limits_{z^* \in \calZ}\, \one\big\{T(z^*,R, X_R, Y_R(\cdot)) \leq T(Z,R, X_R, Y_R(\cdot))\big\}\cdot\one\big\{ G(z^*)=G(Z), S(z^*)=S(Z)\big\}\cdot q(z^*\mid R, X_R, Y_R^*)}{\sum\limits_{z^* \in \calZ}\,\one{\big\{G(z^*)=G(Z), S(z^*)=S(Z)\big\}}\cdot q(z^*\mid R, X_R, Y_R^*)}.
\end{equation}

\end{proposition}
\begin{proof}[Proof Sketch]
    The left-hand side of~\eqref{eq:factorization-new} is proportional to {$\Prob(Z=z, R=r, X_R = x_r, Y_R(\cdot) = y_r(\cdot),S(Z)=s, G(Z)=g)$} for any given $(r, x_r, y_r(\cdot),s,g)$. The key idea of the proof is constructing a factorization of this probability according to the topological ordering of the DAG (or equivalently the progress of time in the experiment):
    \begin{equation*}
        \ldots\quad\rightarrow\quad R_k \quad\rightarrow\quad X_{R_{k}},Y_{R_k}(\cdot) \quad\rightarrow\quad Z_k \quad\rightarrow\quad S_k \quad \rightarrow\quad\ldots
    \end{equation*}
    At each stage $k$, we obtain four terms: the conditional
    probability of $R_k$ given all data from previous stages, the
    conditional probability of $X_{R_{k}},Y_{R_k}(\cdot)$ given
    previous data and $R_k$, and so on. Multiplying the terms from all stages yields a factorization of $\Prob(Z=z, R=r, X_R = x_r, Y_R(\cdot) = y_r(\cdot),S(Z)=s, G(Z)=g)$. Assumptions~\eqref{eq:A2} and~\eqref{eq:A3} help us to simplify the resulting expression such that we can subsume the factors independent of $Z$ into a normalizing constant. This yields~\eqref{eq:factorization-new}. \revd{To derive the formula of the selective randomization p-value, we use a similar factorization which additionally involves the $Z_k^*$ and can be simplified via~\eqref{eq:zstar-id}-\eqref{eq:zstar-sel}.}
    Details of these proofs can be found in the Supplementary Material~\ref{app:proof-computability}.
\end{proof}

\revc{The formula in Proposition~\ref{prop:computability} provides a
  straightforward way of computing (or approximating) the selective
  randomization p-value: if $\calZ$ is discrete and small enough, we
  can directly compute $P_\textnormal{sel}$; otherwise, }
we can approximate it using Monte Carlo methods such as rejection sampling or a Markov Chain Monte Carlo algorithm like the random walk Metropolis-Hastings (RWM) sampler. The details are deferred to Supplementary Material~\ref{sec:computation}. If Assumption~\eqref{eq:A2-star} holds, the expression of $q$ can be simplified as:
\begin{equation*}
    q(z\mid r,x_r,y_r)
    = \prod_{k=1}^K \Prob(Z_k = z_k \mid R_{k} = r_k, X_{R_{k}} = x_{r_k}, S_{k-1} =s_{k-1}).
\end{equation*}
This allows one to parallelize the sampling for the different stages
as the distribution of each $Z_k^*$ depends on data from previous
stages only through $S_{k-1}$ which is conditioned on.

\subsection{Confidence Intervals and Estimation}\label{sec:confint}

Thus far, we have focused on hypothesis tests, but the proposed selective randomization p-value can also be used to construct confidence intervals and estimators for a homogeneous treatment effect. \revc{Defining such an effect is possible if the outcomes are continuous and real-valued and the difference of two potential outcomes is indeed \emph{constant} across (a subgroup of) units. Then, a homogeneous treatment effect~$\tau \in \RR$} can be defined as the difference between the potential outcomes under two treatments $l$ and $l'$ in a subgroup of the recruited units~$I$. This subset of units may be defined in terms of the covariates and can be adaptive, i.e.\ chosen by $S_K$. 
The null hypothesis associated with a homogeneous treatment effect~$\tau$ is given by
\begin{equation*}
    \mathrm{H}_0^\tau\colon\quad Y_i(l)-Y_i(l') = \tau\quad \forall\, i \in I\subseteq R,
\end{equation*}
and let the corresponding selective randomization p-value be denoted as $P_\text{sel}(\tau)$. (Note that we can equally define the homogeneous treatment effect $\tau$ across more than two treatments.) \revc{Setting $\tau =0$ recovers the previous null hypothesis of ``no effect whatsoever'' which is also well-defined for categorical outcomes.} To emphasize the dependence of the selection statistic on $\tau$, we use the notation $S(Z;\tau)$ in the following.

Using the duality of hypothesis testing and confidence sets, we can simply invert the selective hypothesis tests for different values of $\tau$ to construct a $(1-\alpha)$-confidence set as
\begin{equation}\label{eq:confidence-set}
        C_{1-\alpha} := \{\tau \in \RR\colon P_\text{sel}(\tau) > \alpha\}.
\end{equation}
Applying a similar reasoning, we can estimate the homogeneous effect as the $\tau$-value that achieves a selective randomization p-value of $1/2$, i.e.\ $\hat{\tau} = \tau$ such that $P_\text{sel}(\tau) = 1/2$. More specifically, since it may be impossible to find a value of $\tau$ that fulfills this condition exactly, the Hodges-Lehmann (HL) estimator is defined as
\begin{equation*}\label{eq:hl-estimator}
    \hat{\tau} = \frac{\sup\{\tau\colon P_\text{sel}(\tau)< 1/2\}+\inf\{\tau\colon P_\text{sel}(\tau)> 1/2\}}{2}.
\end{equation*}
\revc{For an investigation of its properties, we refer to \textcite{jr_estimates_1963, rosenbaum_hodges-lehmann_1993}.}

When the p-value $P_\text{sel}(\tau)$ is monotonically increasing in $\tau$, the curve $\tau \mapsto P_\text{sel}(\tau)$ crosses every level exactly once. Consequently, the proposed confidence set and estimator are ``well-behaved'': $C_{1-\alpha}$ as defined in \eqref{eq:confidence-set} is an \emph{interval} and the HL estimator provides a sensible point estimate. For unconditional randomization inference, \textcite{caughey_randomisation_2023} state conditions under which the corresponding p-value is indeed monotone.

However, when we condition on the selection event, the p-value curve
is generally not monotone, not even under the conditions in \textcite{caughey_randomisation_2023}. 
\revc{To explain this phenomenon, consider the support of the selective randomization distribution
\begin{equation*}
    \calZ_{S,G}(\tau) := \{z^* \in \calZ \colon S(z^*;\tau) = S(Z;\tau), G(z^*)= G(Z)\}.
\end{equation*}
In many settings, the number of feasible treatment assignments $\lvert \calZ_{S,G}(\tau)\rvert$ becomes very small as $\tau \to \pm\,\infty$. In Example~\ref{ex:enrichment}, this may for instance occur if we select the group with high genetic risk score in the realized trial ($S(Z;\tau) = \text{only high}$) and test a null hypothesis for a very small treatment effect $\tau \ll0$. In this case, there are few alternative treatment assignments $z^*$ that lead to choosing the same group ($S(z^*;\tau) = \text{only high}$) because a small~$\tau$ indicates that the treatment is more effective in the group with low genetic risk. For this reason, there are not many summands in the numerator and denominator of the $P_\textnormal{sel}$-formula~\eqref{eq:p-sel-formula}. This has two consequences: (i) the selective randomization p-value can assume only a comparatively small number of values which complicates rejecting at small $\alpha$-levels; (ii) $P_\textnormal{sel}$ becomes unstable in the sense that small changes in $\tau$ can lead to big differences in the corresponding p-values. 
}
Therefore, the p-value curve $\tau \mapsto P_\text{sel}(\tau)$ becomes increasingly jagged as $\tau \to \pm\,\infty$. As a consequence, the confidence set defined via inversion of tests may be disconnected and the p-value curve may cross any significance level multiple times, which complicates the definition of the HL estimator and computation of the confidence set. We demonstrate this phenomenon with the two-stage trial example~\ref{ex:enrichment}. More details on this simulation study can be found in Section~\ref{sec:simulation-study} and Supplementary Material~\ref{app:details-hold-out}. When the selection of the null hypothesis depends on data from the first \emph{and} second stage, the first panel of Figure~\ref{fig:hold-out} shows that the p-value curve may be far from monotone. For instance, it crosses the level $1/2$ six times yielding six potential point estimates.

\begin{figure}
        \centering
        \includegraphics[scale=0.6]{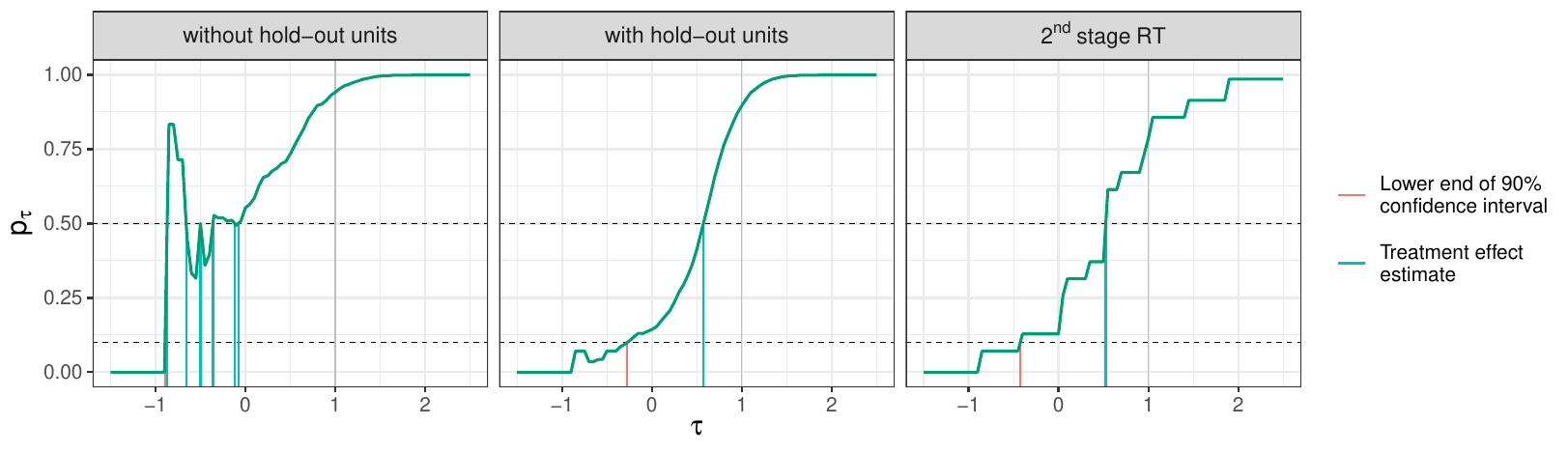}
        \caption{Smoothing effect of hold-out units. In the first panel, the p-value curve (green) is plotted for a selection rule that uses stage 1 and 2 units; in the second panel, only stage 1 units are used for selection; in the third panel, only stage 1 units are used for selection but inference is based on the second stage randomization p-value, cf.\ data splitting. The true treatment effect equals 1 in this simulation.}
    \label{fig:hold-out}
\end{figure}

\revc{To mitigate the risk of a non-monotonic p-value curve and the associated problems, we can use a \emph{hold-out} approach. That is to say, we use selection statistics that do not depend on the treatment assignments $z^*_{\text{ho}}$ to the so-called hold-out units.} For example, the experimenters may decide on the null hypothesis before
gathering data from a last batch of study participants \revc{or they only use a certain percentage of the units at each stage to choose the design.} This idea is closely related to the concept of data carving in the selective inference literature \parencite{fithian_optimal_2017,panigrahi_carving_2023}. \revc{The support of the selective randomization distribution can then be expressed as
\begin{equation*}
    \calZ_{S,G}^{\text{ho}}(\tau) = \{(z^*_{\text{ad}},z^*_{\text{ho}}) \in \calZ \colon S(z^*_{\text{ad}};\tau) = S(Z_{\text{ad}};\tau), G(z^*_{\text{ad}},z^*_{\text{ho}})= G(Z_{\text{ad}},Z_{\text{ho}})\},
\end{equation*}
where $z^*_{\text{ad}}$ denotes the treatment assignments to the adaptive (i.e.\ not hold-out) units.} Since $S$ does not depend on the treatments assigned to hold-out units, the number of alternative treatments $z^*$ cannot sink below a certain threshold even for extreme values of~$\tau$:
\revc{
\begin{equation*}
    \big\lvert \calZ_{S,G}^{\text{ho}}(\tau) \big\rvert \,\,\geq\,\, \big\lvert
    \{ (Z_{\text{sel}}, z^*_{\text{ho}}) \colon
    S(Z_{\text{sel}};\tau) = S(Z_{\text{sel}};\tau),\
    G(Z_{\text{sel}}, z^*_{\text{ho}}) = G(Z_{\text{sel}}, Z_{\text{ho}}) \} \big\rvert.
\end{equation*}
}

In the second panel of Figure~\ref{fig:hold-out}, we plot the p-value curve when only the first stage data is used for selection of the null hypothesis and the second stage participants serve as hold-out units. The curve is still not monotone but clearly smoother. It yields a well-defined HL estimator and the $90\%$ confidence set is a simple interval.

In order to compute the p-value curve and thus obtain an estimate or confidence interval for $\tau$, we typically use Monte Carlo methods like rejection sampling or the RWM sampler. The Supplementary Material~\ref{sec:computation} contains the details on these methods and Subsection~\ref{sec:compute.confidence.interval} discusses potential shortcuts for computing confidence intervals.

\section{Numerical Experiments}\label{sec:empirical}

We use simulations and hypothetical datasets generated from a real
clinical trial to compare three types of randomization tests: the
naive randomization test (RT) given in~\eqref{eq:naive-rand-p-value},
which does not account for the adaptive design, the data-splitting
randomization test (RT 2nd) in~\eqref{eq:2nd-stage-p-value}, which
only uses the second stage data for inference, as well as the proposed
selective randomization test (SRT) in \eqref{eq:selective-p-value} and
\eqref{eq:rand-p-value}. 
The naive and data-splitting randomization p-values can be computed
via a simple Monte Carlo approximation as sampling from the
corresponding randomization distribution is straightforward. For the
selective randomization p-value, we compare two Monte Carlo schemes:
rejection sampling (RS) and the random-walk Metropolis-Hastings (RWM)
sampler.


The
source code of the numerical experiments is available in the GitHub
repository
\url{https://github.com/ZijunGao/Selective-Randomization-Inference-for-Adaptive-Experiments}. 
Some additional numerical results can be found in Supplementary Material~\ref{app:additional-results-num-exp}, including a comparison of different step sizes for RWM.


\subsection{Simulation Study}\label{sec:simulation-study}

To validate our proposed method, we use a simulation study motivated by the hypothetical two-stage enrichment trial in Section~\ref{sec:introduction}. \revc{Note that the design and null hypothesis only depend on data from the first stage; hence, the units in the second stage serve as hold-out units.} We use the following notations and data-generating mechanism. The population of interest is subdivided into two groups with low and high genetic risk score, respectively; this is captured in the covariate $X_i \in \{\text{low}, \text{high}\}$, where~$i$ is the patient index. Every participant is assigned to the new drug ($Z_i=1$) or the placebo ($Z_i=0$).
In both stages of the trial, we use a
completely randomized treatment assignment mechanism. That is,
we assign half of the units to the new drug and half of the units to the placebo with equal probability. The outcome variable of interest $Y_i$ is continuous (e.g.\ cholesterol level). 
In this simulation study, we generate the potential outcomes $Y_i(1) = Y_i(0)$ as i.i.d.\ standard normals. Hence, there is no treatment effect in either group and the sharp null hypothesis holds true. In the first stage, we simulate $n_1=100$ recruited patients, 50 with low and high genetic risk score, respectively. To select the group(s) to recruit from in the second stage, we compare the standardized average treatment effects $\Delta_x$ in the first stage, which are defined as
\begin{equation}\label{eq:standardized-ate}
    \Delta_x = \frac{\bar{Y}_{x,1} - \bar{Y}_{x,0}}{\sqrt{\hat{\sigma}^2_{x,1} + \hat{\sigma}^2_{x,0}}},\qquad
    \hat{\sigma}^2_{x,z} = \frac{
    \sum_{i=1}^{n_1} \one{\{X_i=x, Z_i=z\}}\,
    (Y_i-\bar{Y}_{x,z})^2}{\sum_{i=1}^{n_1}
    \one{\{X_i=x, Z_i=z\}}},\qquad
    x \in \{\text{low}, \text{high}\},\, z \in\{0,1\},
\end{equation}
where $\bar{Y}_{x,z}$ denotes the average outcome in group $x$ under treatment $z$. We define the selection statistic $S$ in terms of the scaled ATE difference $\Delta = (\Delta_{\text{high}} - \Delta_{\text{low}})/\sqrt{2}$ as follows
\begin{equation}\label{eq:selection-rule}
    S = \begin{cases}
        \text{only low}, &\quad \Delta < \Phi^{-1}(0.2),\qquad \text{recruit 40 from group $X_i=\text{low}$},\\
        \text{only high}, &\quad \Delta > \Phi^{-1}(0.8),\qquad \text{recruit 40 from group $X_i=\text{high}$},\\
        \text{both}, &\quad \text{otherwise},
        \hspace{0.62cm}
        \qquad\,
        \text{recruit 20 from group $X_i=\text{low}$ and $X_i=\text{high}$ each}.
    \end{cases}
\end{equation}
Here $\Phi^{-1}(\cdot)$ denotes the quantile function (inverse cumulative distribution function) of the standard normal distribution. This selection rule is motivated by the fact that if the distribution of the outcomes is the same under placebo and drug for both groups and the data points are i.i.d., then $\Delta$ approximately follows a standard normal distribution. Note that these assumptions are only used to design the experiment but not to analyse the data.

For the test statistic $T$, we use the standardized ATE estimator
in~\eqref{eq:standardized-ate} applied to the selected group(s) with data from both stages; \textcolor{black}{a discussion of test statistic selection, particularly adaptive test
statistics, is provided in Supplementary Material~\ref{app:test.statistic}
}.
A larger value of the test statistic $T$ indicates a stronger treatment effect.
We perform one-sided hypothesis tests on null hypotheses across a range of treatment effects $\tau \in \{-1, -0.8, \ldots, 0.8, 1\}$, where smaller treatment effects should have a higher probability of being rejected.
\textcolor{black}{Each p-value is based on $400$ Monte Carlo samples.}
We set the significance level to $\alpha=0.1$ and generate 400 datasets from the data generating mechanism.

In \Cref{fig:default}, we report the empirical rejection probabilities for each null hypothesis and also display them according to the subgroup selected in stage 1: subgroup $X_i=\text{low}$, subgroup $X_i=\text{high}$ or both subgroups. SRT and RT 2nd both control the selective type-I errors.
\textcolor{black}{Although SRT is not always more powerful than data splitting (RT 2nd), as the two tests are applied to different subsets of units, SRT leverages the additional information from the first stage and therefore usually achieves higher power, especially when both subgroups are selected.} Moreover, the power curves do not significantly differ across rejection sampling and RWM. 

In addition, we compare the computation time in a personal desktop for
a single 
RT, RT 2nd, SRT with rejection sampling, and SRT with RWM sampling. RT
and RT 2nd are the most efficient, followed by RWM. Rejection sampling
is significantly slower, especially when only one of the subgroups is
selected and the hypothesis being tested is extreme. In this scenario,
the rejection probability in RS is very high, leading to a
considerable computational burden.




\begin{figure}[t]
        \centering
    \begin{minipage}{0.96\textwidth}
                \centering
                \includegraphics[clip, trim = 0.5cm 0cm 0cm 0cm, width = \textwidth]{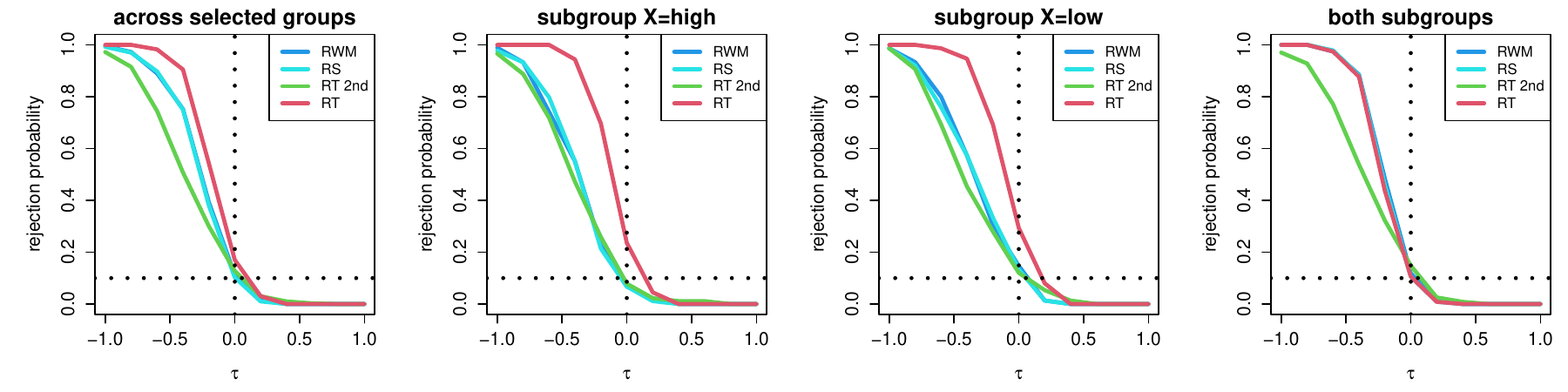}
                \subcaption[(a)]{Rejection probability.}
        \end{minipage}
        \begin{minipage}{0.96\textwidth}
            \centering
            \includegraphics[clip, trim = 0.5cm 0cm 0cm 0cm, width = \textwidth]{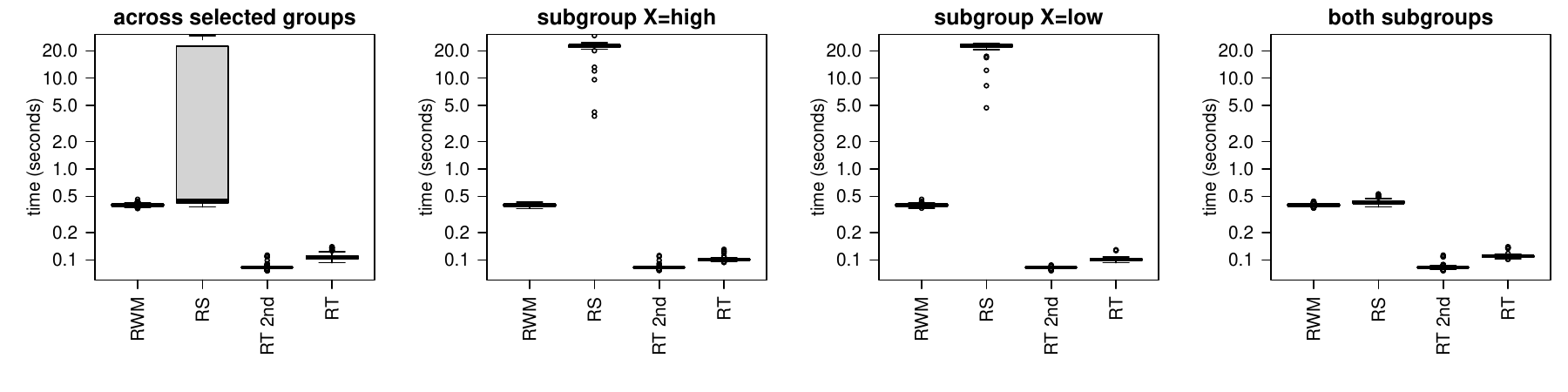}
            \subcaption[(b)]{Computation time.}
        \end{minipage}
        \caption{
        Comparison of randomization tests regarding rejection probabilities and computation time for testing a sequence of nulls of constant treatment effect $\tau$. The potential outcomes satisfy $\tau = 0$.
        The leftmost column reports the rejection probability and computation time averaged over all subgroup selections; the right three columns report the results among trials with a specific subgroup selection: selecting only the subgroup with high $X_i$ ($S = \text{only high}$), only the subgroup with low $X_i$ ($S = \text{only low}$), or both subgroups with high or low $X_i$ ($S = \text{both}$).
        \textcolor{black}{The results are aggregated over $400$ repeats, and the standard deviation of the rejection probabilities is bounded above by $0.025$.}}
    \label{fig:default}
\end{figure}







Lastly, we compare the confidence intervals constructed via RT, RT 2nd and SRT.
Since the p-value curve $\tau \mapsto P_\text{sel}(\tau)$ is not guaranteed to be monotone, we resort to simply specifying an interval of treatment effect, discretizing it and conducting a hypothesis test for each resulting bin.
We construct one-sided confidence intervals of the form $[\underline{\tau}, \infty)$ and evaluate their coverage.
Furthermore, we compute the average value of $\underline{\tau}$, where
a greater value of $\underline{\tau}$ indicates higher
power. Consistent with the previous results, the confidence intervals
based on RT achieved an empirical coverage of $85.8\%$ when the
nominal level is $90\%$. Particularly when only one of the subgroups is selected, RT
has an undercoverage of more than 10\%. On the contrary, SRT and RT
2nd achieve the nominal coverage rate. Moreover, the lower confidence bound
obtained by SRT is in general larger than that of RT 2nd, which again
shows the SRT is more powerful than RT 2nd. The detailed results can
be found in Supplementary Material~\ref{app:additional-results-num-exp}. 




\subsection{Hypothetical Real Data Analysis}\label{sec:real.data}

\begin{table}[t]
\caption{Randomization tests on a hypothetical dataset generated from the SPRINT trial data.}
\label{tab:overall}

\begin{subtable}[t]{0.65\textwidth}
\caption{Data summary and p-values of Fisher's randomization test per subgroup for stage 1 and 2.}
\label{tab:real.data}
\centering
\begin{tabular}[t]{cc|rrrr|r}
\toprule
\multicolumn{2}{c|}{Stage}                     & \multicolumn{4}{c|}{Stage 1} & Stage 2 \\ 
\multicolumn{2}{c|}{Age group}                     &  $\le 60$   & $[60, 69]$    &  $[70, 79]$   &  $\ge 80$   & $\ge 80$ \\ \midrule
Control & $Y = 1$ &  8   &  17   &   22  &   19  &  17 \\ 
&  $Y = 0$ &  187   &   349  &   275  &   123  & 87 \\ 
Treated &  $Y = 1$ &  12   &  13   &   22  &   7  & 13 \\ 
&  $Y = 0$ &  201   &   331  &   289  &  125   & 83 \\
\midrule
\multicolumn{2}{c|}{Fisher's exact test} & 0.823 & 0.307 & 0.460 & 0.005 & 0.314
\\\bottomrule
\end{tabular}
\end{subtable}
\hspace{0.5cm}
\begin{subtable}[t]{0.25\textwidth}
\caption{P-values for the null hypothesis of no treatment effect in the age group $\geq 80$. 
}
\label{tab:real.data.pvalues}
\centering
\begin{tabular}[t]{c|c}
\toprule
     &  p-value\\\midrule
     SRT & 0.048 \\
     RT 2nd & 0.314 \\
     RT & 0.027\\\bottomrule
\end{tabular}
\end{subtable}
\end{table}

Next, we create a hypothetical two-stage enrichment trial from a real one-stage randomized experiment and compare RT, RT 2nd and SRT.
The hypothetical dataset is generated from the Systolic Blood Pressure Intervention Trial (SPRINT) \parencite{sprint2016systolic, gao2021assessment} which studies whether a new treatment program for reducing systolic blood pressure will lower the risk of a cardiovascular disease (CVD).
The primary clinical outcome is the occurrence of a major CVD event.
We classify the subjects into four age subgroups: $\leq 59$, $60$ to
$69$, $70$ to $79$, and $\geq 80$.
For stage~1, we first randomly sample $2000$ units without replacement, then estimate the ratio of the probability of a CVD event in the treatment and control arm (i.e.\ relative risk) 
for each age group, and lastly select the age group with the smallest relative risk, i.e.\ the age group with the most significant improvement. For stage~2, we randomly sample an additional $200$ units without replacement from the units of the selected age group that were not already part of the first stage.
An instance of this simulated dataset is summarized in
\Cref{tab:real.data}, where the $\geq 80$ age group is selected after
stage~1.

We test the sharp null that there is no treatment effect in the selected group. To this end, we use the relative risk as the test statistic and approximate the p-value using $1000$ Monte Carlo samples.
For the selective randomization test, we use the rejection sampling algorithm.
As shown in \Cref{tab:real.data.pvalues}, in this instance SRT reports
a smaller p-value than RT 2nd. The naive randomization test (RT),
which does not control the selective type-I error, gives an even
smaller p-value.

\begin{figure}[t]
    \centering
    \begin{subfigure}[t]{0.45\textwidth}
         \centering
        \includegraphics[clip, trim = 0.5cm 0cm 0cm 0cm, width = \textwidth]{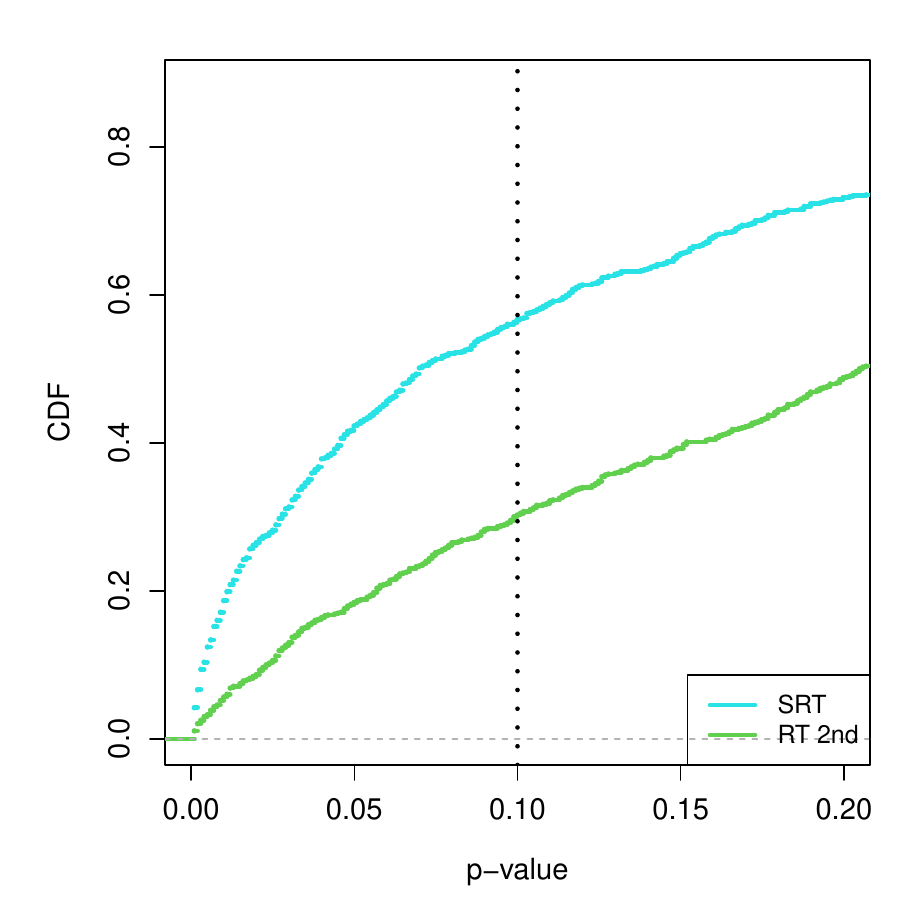}
         \caption{Power: CDF of SRT and RT 2nd p-values.}
         \label{fig:real.data}
     \end{subfigure}
     \hfill
     \begin{subfigure}[t]{0.45\textwidth}
         \centering
         \includegraphics[clip, trim = 0.5cm 0cm 0cm 0cm, width = \textwidth]{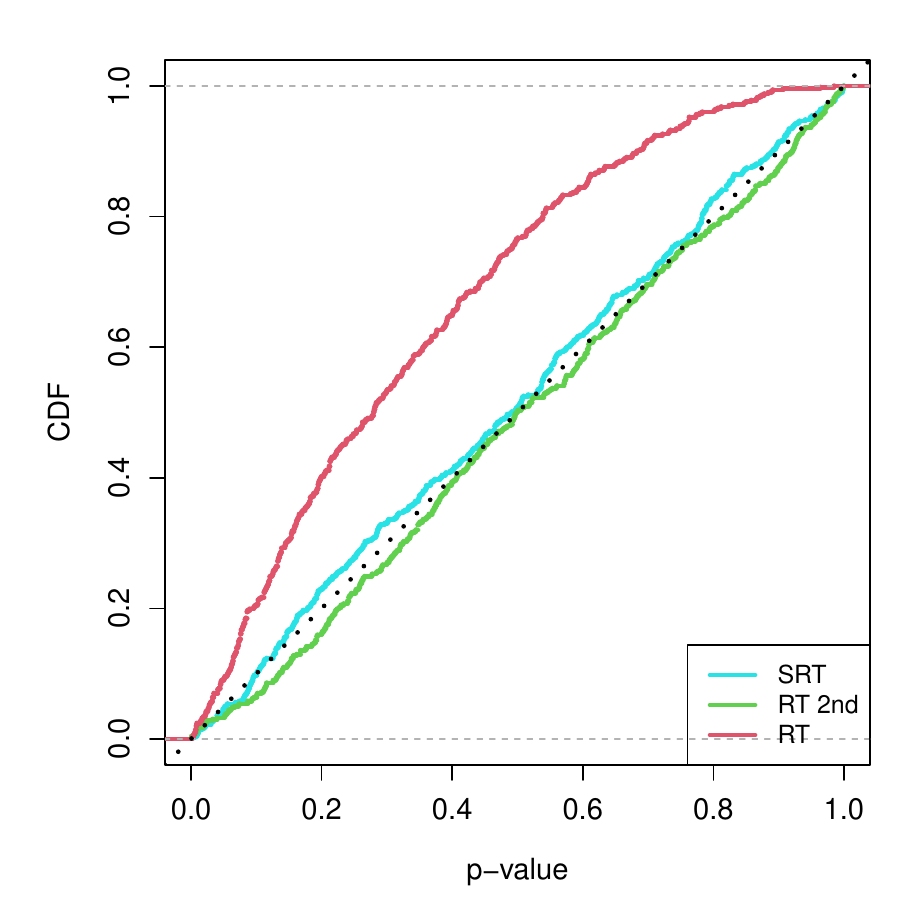}
         \caption{Type-I error: CDF of SRT, RT 2nd and RT p-value in the placebo analysis.}
         \label{fig:real.data.placebo}
     \end{subfigure}
     \caption{Analysis of power and type-I error control in hypothetical datasets generated from the SPRINT trial.}
\end{figure}

To compare the different randomization tests more systematically, we generate 1000 hypothetical datasets according to the procedure above focusing on the trials where the age group $\geq 80$ was selected. In Figure~\ref{fig:real.data}, we plot the empirical cumulative distribution functions (CDFs) of RT 2nd and SRT. Clearly, the selective randomization p-value stochastically dominates the second stage randomization p-value leading to an increase in power. In particular, SRT rejects the null hypothesis of no treatment effect at level $0.1$ in 56.6\% of the generated datasets, whereas RT 2nd only rejects this null hypothesis in 30.2\% of the datasets.

To evaluate type-I error control, we perform a placebo analysis: We
generate multiple ``negative control'' datasets by only sampling
units from the control group of the SPRINT trial.  As before, we focus
on datasets where the age group $\geq 80$
was selected. Instead of using the treatments assigned in the original
study, we generate ``fake treatments'' according to independent fair
coin flips. The rest of the design is kept the same.
Since all the units in the
constructed datasets were part of the SPRINT control group, they did
not receive the treatment and the null hypothesis that the artificial
treatment has no effect is true. Hence, p-values that control the
selective type-I error have a uniform distribution. In
Figure~\ref{fig:real.data.placebo}, we plot the empirical CDFs of the
different randomization p-values: While SRT and RT 2nd are close to
the diagonal line and thus are basically uniformly distributed, the
CDF of RT is stochastically larger, again showing that the naive
randomization test does not control the selective type-I error.



\section{Discussion and Outlook}\label{sec:discussion}


To analyse data from adaptive studies, we use the paradigm of
randomization inference due to its two main advantages. First,
randomization inference does not require any modelling assumptions on
the distribution of the covariates and outcomes or independence of the
data points. Therefore, it is an ideal default test of the null
hypothesis that can complement a further model-based analysis that
relies on additional assumptions. Second, randomization inference also
facilitates post-selection inference. Our proposed selective
randomization p-value can accommodate a very wide range of adaptive
designs due to a simple observation: in standard randomization
inference, we compare the realized treatment assignment $Z$ to all
alternative assignments $Z^*$; in adaptive studies, however, it is
fair to only compare $Z$ to $Z^*$ that would lead to the same design
choices. In other words, it suffices to condition on $S(Z^*)=S(Z)$
regardless of how the selection statistic $S$ changes the design of
the future stages. \revc{
Crucially, this approach controls the selective (or conditional) type-I error in contrast to the unconditional type-I error. That is, the inference is valid given the data from the realized experiment rather than merely controlling the type-I error on average across all possible realizations of the experiment.

Selective inference has been also used in the commonplace super-population framework with i.i.d.\ outcomes in order to analyse adaptive experiments. \textcite{chen_optimal_2023} and \textcite{andrews2024inference} consider parametric models and derive asymptotically valid methods with and without adaptive treatment assignment distribution, respectively. These works rely on arguably strong assumptions on the parameter of interest, selection procedure and treatment assignment distribution. In our article, we specifically chose to use randomization inference 
instead to minimize such modelling assumptions. Thus, we provide selective inference methods for adaptive experiments in the finite population framework which complement the previous super-population tools.

\textcite{nair2023weightedmc} take a different path and avoid strong parametric assumptions via a weighted resampling test while still operating in the super-population framework.
Although their p-value bears similarities to our selective randomization p-value~\eqref{eq:p-sel-formula}, the respective interpretations are quite different. Most importantly, \citeauthor{nair2023weightedmc} only consider pre-specified recruitment and null hypothesis; hence, a selective adjustment is not necessary in their setting. Furthermore, our work treats \emph{potential} outcomes as fixed and formulates the null hypothesis as (deterministic) contrasts between them. \textcite{nair2023weightedmc}, on the other hand, treat the \emph{realized} outcomes as stochastic and aim to test probabilistic null hypotheses, like stationarity of the outcomes or conditional independence of the outcomes and the treatment.

}


While this article is mainly concerned with introducing selective
randomization inference and addressing immediate questions such as
computation of  confidence intervals, there are many conceivable
extensions.

First, this work develops a general theoretical framework to
\emph{analyse} data from adaptive experiments. In turn, the developed methodology also has implications on the \emph{design} of such adaptive studies. Specifically, the preliminary data analyses at each stage and how they inform the design of the next stage must obey Assumptions~\eqref{eq:A2} and~\eqref{eq:A3} or their stronger versions~\eqref{eq:A2-star} and~\eqref{eq:A3-star}. \revc{Beyond fulfilling these minimal requirements, it is of great practical interest how to choose the adaptive design (and thus $S$) such that the resulting inference remains powerful and other criteria like patient benefit are still achieved. This is particularly important when experimenters use the hold-out approach and need to decide how many units (and at which stages) are not used for informing the design. 
}
\textcolor{black}{
Similarly, it is of practical relevance to study the power of test statistics, as their choice can significantly affect the usefulness of our proposal \parencite{zhang2025powerful}. In particular, since the test statistic can depend arbitrarily on the selection rule without compromising validity of our proposal, one future direction of interest is to characterize the most powerful selection-dependent test statistics.
}

Second, our work only considers the classical problem of testing a (partially) sharp null hypothesis. In constructing confidence intervals, we further assume the treatment effect is the same for every  individual (in a subgroup). 
Many recent works attempt to relax these assumptions and broaden the scope of randomization inference. For example, \textcite{caughey_randomisation_2023} develops a randomization test for quantiles of the individual treatment effect. Another line of research considers weaker null hypotheses about the average treatment effect among the participants \parencite{ding_paradox_2017, wu_randomization_2021, cohen_gaussian_2022}. Moreover, \emph{conditional} randomization inference has recently been used to analyse observational data in order to account for the matching of units \parencite{pimentel_covariate-adaptive_2024,pimentel2024}. Some of the computational contributions may also help to better approximate the selective randomization p-value.


Third, in some studies it may be desirable to use the first stage data to determine a threshold for a continuous biomarker and only recruit patients that pass that threshold in the second stage \parencite{friede2012conditional, stallard2014adaptive, rosenblum2016group, lai2019adaptive, stallard_adaptive_2023}. For example, in our running example, we may wish to use the first stage to select a cut-off for the genetic risk score. In this case, the selective randomization p-value may reduce to the data splitting p-value, because the selected cut-off may be uniquely determined by the first-stage treatment assignment. Therefore, it would be interesting to investigate whether we can retain some information from the first stage data by using special selection rules \textcolor{black}{\parencite{gao_selective_2025}} and making additional assumptions about treatment effect heterogeneity.
Beyond this extension, it would also be of interest to consider the occurrence of drop-out units, i.e.\ participants with missing outcomes, in adaptive experiments. Solutions to this problem include model-based or non-informative outcome imputation as well as reweighting the treatment assignment probabilities of the units with observed outcomes, e.g.\ \textcite{edgington2007randomization, rosenberger_randomization_2019, ivanova2022randomization, heussen2023randomization}.

\section*{Funding and Acknowledgements}
TF was supported by a PhD studentship from GlaxoSmithKline Research \&
Development Limited. QZ was partly supported by the Engineering and
Physical Sciences Research Council (grant number EP/V049968/1).

\clearpage
\appendix
\renewcommand{\thesection}{S\arabic{section}}
\renewcommand{\theequation}{\thesection.\arabic{equation}}
\renewcommand{\thefigure}{S\arabic{figure}}
\renewcommand{\thetable}{S\arabic{table}}
\renewcommand{\theproposition}{S\arabic{proposition}}
\renewcommand{\thedefinition}{S\arabic{definition}}
\renewcommand{\thelemma}{S\arabic{lemma}}
\renewcommand{\thecorollary}{S\arabic{corollary}}
\renewcommand{\thetheorem}{S\arabic{theorem}}

\section*{Supplementary Material}
\addcontentsline{toc}{section}{Supplementary Material}


\section{Imputability}\label{supp:computability}
\revc{
As explained in Section~\ref{sec:computability}, it is not obvious if the selective randomization p-value can be computed when the null hypothesis is only partially sharp.
Here, we elaborate in more detail on properties of $T$, $S$ and $q$ as well as the conditioning function $G$ that render former imputable. First, we state an intuitive property of quantities depending on the potential outcomes~$Y_R(\cdot)$ and the treatment~$Z$.
\begin{definition}
    {A function $f$ of $Z,R,X_R$ and $Y_R(\cdot)$ is called \emph{single-world} if $f(z,R,X_R,Y_R(\cdot)) = f(z,R,X_R,Y_R(z))$ for all $z\in \calZ$.}
\end{definition}
In essence, this means that $f$ can only depend on outcomes that are observed under the treatment $Z$ or, in other words, $f$ is independent of potential outcomes that would be realized under different treatment assignments (alternative worlds). The selection statistics $S_1,\ldots,S_K$ and treatment assignment distributions $q_1,\ldots,q_K$ are single-world objects by necessity; otherwise, they could not be evaluated when conducting the adaptive experiment. (For this reason, we defined them as functions of $Y_R$ instead of $Y_R(\cdot)$ in Section~\ref{sec:notation}.) Going forward, we also assume that the test statistic~$T$ is a single-world function. While not strictly necessary, imposing this condition is very intuitive. For instance, if $T$ is not single-world, even its value under the \emph{realized} treatment $Z$ would depend on unobserved potential outcomes.

To conduct randomization inference, we need a (partially) sharp null hypothesis that lets us impute unobserved potential outcomes from the observed ones $Y_R$. For stage $k$ with realized treatment assignment $Z_k$, we define the set of all imputable potential outcomes~as
\begin{equation*}
    Y_{R_k}^\text{imp}(\cdot) := \{Y_{R_k}(z_k^*)\colon\, z_k^* \in \calZ_k, Y_{R_k}(Z_{k})=Y_{R_k}(z_{k}^*) \text{ under } \mathrm{H}_0\} \quad \subseteq\quad Y_{R_k}(\cdot).
\end{equation*}
Analogously to the previous notation, we denote the imputable potential outcomes up to the $k$-th stage $Y_{R_{[k]}}^\text{imp}(\cdot)$ and all imputable potential outcomes $Y_{R}^\text{imp}(\cdot)$. Under Fisher's sharp null hypothesis, all potential outcomes $Y_R(\cdot)$ can be imputed, i.e.\ $Y_{R}^\text{imp}(\cdot) = Y_R(\cdot)$, and the randomization p-value can be computed. Under a partially sharp null hypothesis $Y_{R}^\text{imp}(\cdot)$ is a proper subset of $Y_R(\cdot)$, however. To ensure that we can still evaluate $S$, $q$ and $T$, we restrict the support of the treatment assignment distribution via conditioning on $G = G(Z)$ so that we only sample treatments $z^*$ whose pertaining outcomes $Y_R(z^*)$ can be imputed. We formalize this idea in the following definition.

\begin{definition}\label{def:imputability}
    For $k \in [K]$, let $G_k$ be a function of $Z_k, R_k$ and $X_{R_k}$; denote $G = (G_1, \ldots, G_K)$ and $G_{[k]} = (G_1, \ldots, G_k)$. Suppressing the dependence on $R$ and $X_R$ in the notation, we define \emph{imputability} for $T$, $q_k$ and $S_k$ under the null hypothesis with respect to $G$ as follows:
    \begin{enumerate}
        \item Test statistic: $T$ is imputable if $Y_{R}(z^*) \in Y_R^\text{imp}(\cdot)$ for all $z^* \in \calZ$ such that $G(z^*) = G(Z)$ and $S(z^*) = S(Z)$.
        \item Treatment assignments: $q_k$ is imputable if $Y_{R_{[k-1]}}(z_{[k-1]}^*) \in Y_{R_{[k-1]}}^\text{imp}(\cdot)$ for all $z_{[k-1]}^* \in \calZ_{[k-1]}$ such that $G_{[k-1]}(z_{[k-1]}^*) = G_{[k-1]}(Z_{[k-1]})$ and $S_{[k-1]}(z_{[k-1]}^*) = S_{[k-1]}(Z_{[k-1]})$.
        \item Selection statistics: $S_k$ is imputable if $Y_{R_{[k]}}(z_{[k]}^*) \in Y_{R_{[k]}}^\text{imp}(\cdot)$ for all $z_{[k]}^* \in \calZ_{[k]}$ such that $G_{[k]}(z_{[k]}^*) = G_{[k]}(Z_{[k]})$ and $S_{[k-1]}(z_{[k-1]}^*) = S_{[k-1]}(Z_{[k-1]})$.
    \end{enumerate}
    If the conditions above are fulfilled for every $k$, we call $T$, $q$ and $S$ \emph{imputable under the null hypothesis with respect to $G$}.
\end{definition}


In Definition~\ref{def:imputability}, we state conditions for imputability in a very general manner. In practice, however, it is usually quite intuitive how to choose a conditioning function $G$ such that $T$, $q$ and $S$ are imputable. Consider for instance the two-stage enrichment design in Example~\ref{ex:enrichment}: If we recruit participants with both low and high genetic risk score in the second stage ($S_1 = \text{both}$), the corresponding null hypothesis is sharp, therefore all potential outcomes are imputable and we do not need to condition on any $G$. If only people with high genetic risk score are recruited in the second stage ($S_1 = \text{only high}$), we cannot impute unobserved potential outcomes from units with low genetic risk in the first stage. Therefore, we need to ensure that alternative treatments~$z^*_1$ agree with the realized values~$Z_1$ for those units.
This is accomplished by conditioning on the function $G_1(Z_1,R_1,X_{R_1}) = (Z_{1,i}\one\{X_{i} = \text{low}\})_{i \in R_1}$. In fact, we can set $G(Z) = G_1(Z_1)$ as the null hypothesis is sharp for outcomes in the second stage. Now it is easy to verify that all imputability conditions in Definition~\ref{def:imputability} are indeed satisfied for the two-stage enrichment trial. The case of recruiting people with low genetic risk score ($S_1 = \text{only low}$) works analogously.
}

\section{Proof of Proposition~\ref{prop:computability}} \label{app:proof-computability}

First, we derive the factorization of the selective randomization distribution. \revc{For ease of exposition, we consider events where each variable takes one specific value; for a more rigorous version of the proof, one may replace them with arbitrary events of the $\sigma$-field generated by the random variables. Let ${z \in \calZ}, s \in \calS, g \in \mathcal{G}$ and let $r$, $x_r$ and $y_r(\cdot)$ denote arbitrary values of $R$, $X_R$ and $Y_R(\cdot)$, respectively.}
We repeatedly use the definition of conditional independence to construct a factorization where each random variable/vector is only conditioned on variables that predate it in the topological ordering of the DAG (and in time):
\begin{align*}
    &\Prob(Z = z\mid R = r, X_R=x_r,Y_R(\cdot) = y_r(\cdot), S(Z)=s, G(Z)=g)\\
    &= \frac{\Prob(Z = z, R = r, X_R=x_r,Y_R(\cdot) = y_r(\cdot), S(Z)=s, G(Z)=g)}{\Prob(R = r, X_R=x_r,Y_R(\cdot) = y_r(\cdot), S(Z)=s, G(Z)=g)} \\
    &\propto \begin{aligned}[t]
    \Prob(G(Z) &= g\mid Z=z, R = r, X_R=x_r,Y_R(\cdot) = y_r(\cdot), S(Z)=s)\,\cdot \\[-0.5ex]
    \prod_{k=1}^K\,\, &\Prob(R_k = r_k \mid R_{[k-1]} = r_{[k-1]}, X_{R_{[k-1]}}=x_{r_{[k-1]}},Y_{R_{[k-1]}}(\cdot) = y_{r_{[k-1]}}(\cdot), Z_{[k-1]}=z_{[k-1]}, S_{[k-1]}(Z_{[k-1]}) = s_{[k-1]})\,\cdot \\[-1.5ex]
    &\begin{aligned}
        \Prob(X_{R_k} = x_{r_k}, Y_{R_k}(\cdot) = y_{r_k}(\cdot) \mid\,
    &R_{[k]} = r_{[k]}, X_{R_{[k-1]}}=x_{r_{[k-1]}},Y_{R_{[k-1]}}(\cdot) = y_{r_{[k-1]}}(\cdot),\\ &Z_{[k-1]}=z_{[k-1]}, S_{[k-1]}(Z_{[k-1]}) = s_{[k-1]})\,\cdot
    \end{aligned}\\[1ex]
    &\Prob(Z_k = z_k \mid R_{[k]} = r_{[k]}, X_{R_{[k]}}=x_{r_{[k]}},Y_{R_{[k]}}(\cdot) = y_{r_{[k]}}(\cdot), Z_{[k-1]}=z_{[k-1]}, S_{[k-1]}(Z_{[k-1]}) = s_{[k-1]})\, \cdot \\[1ex]
    &\Prob(S_{k}(Z_{[k]}) = s_k \mid R_{[k]} = r_{[k]}, X_{R_{[k]}}=x_{r_{[k]}},Y_{R_{[k]}}(\cdot) = y_{r_{[k]}}(\cdot), Z_{[k]}=z_{[k]}, S_{[k-1]}(Z_{[k-1]}) = s_{[k-1]}).
    \end{aligned} 
\end{align*}
To simplify this expression, we apply Assumptions~\eqref{eq:A2} and~\eqref{eq:A3}. Moreover, we use the fact that $G$ is a deterministic function of $Z$, $R$ and $X_R$ and $S_k$ is a deterministic function of $R_{[k]}, X_{R_{[k]}}, Y_{R_{[k]}}$ and $Z_{[k]}$, where $k \in [K]$; this is denoted by (D). We obtain
\begin{alignat*}{3}
    &\mathrlap{\Prob(Z = z\mid R = r, X_R=x_r, Y_R(\cdot) = y_r(\cdot), S(Z)=s, G(Z)=g)}\qquad &&&&\\[1ex]
    && &\mathrlap{\propto \one\{G(z)=\,g\}\, \cdot}\qquad && \tag{by
      (D) } \\[-0.5ex]
    &&&& \prod_{k=1}^K\,\, &\Prob(R_k = r_k \mid R_{[k-1]} = r_{[k-1]}, X_{R_{[k-1]}}=x_{r_{[k-1]}},Y_{R_{[k-1]}}(\cdot) = y_{r_{[k-1]}}(\cdot), S_{k-1}(z_{[k-1]}) =
    s_{k-1})\,\cdot \tag{by \eqref{eq:A3}} \\[-1.5ex]
    &&&& &\begin{aligned}
        \Prob(X_{R_k} = x_{r_k}, Y_{R_k}(\cdot) = y_{r_k}(\cdot) \mid\,
    &R_{[k]} = r_{[k]}, X_{R_{[k-1]}}=x_{r_{[k-1]}},Y_{R_{[k-1]}}(\cdot) = y_{r_{[k-1]}}(\cdot),\\ &S_{k-1}(z_{[k-1]}) = s_{k-1})\,
    \cdot
    \end{aligned} \tag{by \eqref{eq:A3}} \\[1.0ex]
    &&&& &\Prob(Z_k = z_k \mid R_{[k]} = r_{[k]}, X_{R_{[k]}} =
    x_{r_{[k]}}, Y_{R_{[k-1]}}(z_{[k-1]}) = y_{r_{[k-1]}}, Z_{[k-1]} =
    z_{[k-1]})\, \cdot \tag{by \eqref{eq:A2}} \\[1ex]
    &&&& &\one{\{S_k(z_{[k]})=s_k\}} \tag{by (D)} \\[1ex]
    && &\mathrlap{\propto \one\{G(z)=\,g, S(z)=s\}\,\cdot}\qquad && \\
    &&&& \prod_{k=1}^K\,\, &\Prob(Z_k = z_k \mid R_{[k]} = r_{[k]}, X_{R_{[k]}} = x_{r_{[k]}}, Y_{R_{[k-1]}}(z_{[k-1]}) = y_{r_{[k-1]}}, Z_{[k-1]} = z_{[k-1]})\\
    && &\mathrlap{\propto \one\{G(z)=\,g, S(z)=s\}\,\cdot q(z\mid r,x_r,y_r).}\qquad &&
\end{alignat*}
If additionally Assumption~\eqref{eq:A2-star} holds, the expression above simplifies further:
    \begin{multline*}
        \Prob(Z = z\mid R=r, X_R=x_r, Y_R(\cdot) = y_r(\cdot), S(Z)=s, G(Z)=g) \\[-1ex]
        \quad\,\propto\, \one{\{G(z)=g, S(z)=s\}} \prod_{k=1}^K \Prob(Z_k = z_k \mid R_k = r_k, X_{R_k} = x_{r_k}, S_{k-1}(z_{[k-1]}) = s_{k-1}).
    \end{multline*}

\revd{In order to derive the formula of the selective randomization p-value, we apply a similar factorization that also resembles the topological order of the variables in the DAG. In this version, however, we also include the alternative treatment assignments $Z^*$:
{\allowdisplaybreaks
\begin{align*}
    &\Prob(Z^* = z^*\mid Z=z, R = r, X_R=x_r,Y_R(\cdot) = y_r(\cdot), S(Z^*)=S(Z), G(Z^*)=G(Z))\\[1ex]
    &= \frac{\Prob(Z^* = z^*, Z=z, R = r, X_R=x_r,Y_R(\cdot) = y_r(\cdot), S(Z^*)=S(Z), G(Z^*)=G(Z))}{\Prob(Z=z,R = r, X_R=x_r,Y_R(\cdot) = y_r(\cdot), S(Z^*)=S(Z), G(Z^*)=G(Z))} \\
    &\propto \begin{aligned}[t]
    \Prob(G(Z^*) &= G(Z)\mid Z^* = z^*, Z=z, R = r, X_R=x_r,Y_R(\cdot) = y_r(\cdot), S(Z^*)=S(Z))\,\cdot \\[-0.5ex]
    \prod_{k=1}^K\,\, &\begin{aligned}[t]
        \Prob(R_k = r_k \mid\, &R_{[k-1]} = r_{[k-1]}, X_{R_{[k-1]}}=x_{r_{[k-1]}},Y_{R_{[k-1]}}(\cdot) = y_{r_{[k-1]}}(\cdot),\\
        &Z^*_{[k-1]}=z^*_{[k-1]},Z_{[k-1]}=z_{[k-1]}, S_{[k-1]}(Z^*_{[k-1]}) = S_{[k-1]}(Z_{[k-1]}))\,\cdot
    \end{aligned}\\[1ex]
    &\begin{aligned}
        \Prob(X_{R_k} = x_{r_k}, Y_{R_k}(\cdot) = y_{r_k}(\cdot) \mid\,
        &R_{[k]} = r_{[k]}, X_{R_{[k-1]}}=x_{r_{[k-1]}},Y_{R_{[k-1]}}(\cdot) = y_{r_{[k-1]}}(\cdot),\\ &Z^*_{[k-1]}=z^*_{[k-1]}, Z_{[k-1]}=z_{[k-1]}, S_{[k-1]}(Z^*_{[k-1]}) = S_{[k-1]}(Z_{[k-1]}))\,\cdot
    \end{aligned}\\[1ex]
    &\begin{aligned}[t]
        \Prob(Z_k = z_k \mid\, &R_{[k]} = r_{[k]}, X_{R_{[k]}}=x_{r_{[k]}},Y_{R_{[k]}}(\cdot) = y_{r_{[k]}}(\cdot),\\
        &Z^*_{[k-1]} = z^*_{[k-1]}, Z_{[k-1]}=z_{[k-1]}, S_{[k-1]}(Z^*_{[k-1]}) = S_{[k-1]}(Z_{[k-1]}))\, \cdot
    \end{aligned}\\[1ex]
    &\begin{aligned}[t]
        \Prob(Z^*_k = z^*_k \mid\, &R_{[k]} = r_{[k]}, X_{R_{[k]}}=x_{r_{[k]}},Y_{R_{[k]}}(\cdot) = y_{r_{[k]}}(\cdot),\\
        &Z^*_{[k-1]} = z^*_{[k-1]}, Z_{[k]}=z_{[k]}, S_{[k-1]}(Z^*_{[k-1]}) = S_{[k-1]}(Z_{[k-1]}))\, \cdot
    \end{aligned}\\[1ex]
    &\begin{aligned}
        \Prob(S_{k}(Z^*_{[k]}) = S_{k}(Z_{[k]}) \mid\, &R_{[k]} = r_{[k]}, X_{R_{[k]}}=x_{r_{[k]}},Y_{R_{[k]}}(\cdot) = y_{r_{[k]}}(\cdot),\\
        &Z^*_{[k]}=z^*_{[k]}, Z_{[k]}=z_{[k]}, S_{[k-1]}(Z^*_{[k-1]}) = S_{[k-1]}(Z_{[k-1]})).
    \end{aligned}
    \end{aligned} 
\end{align*}
This expression can be simplified via the conditional independences~\eqref{eq:zstar-indep} and~\eqref{eq:zstar-sel}; moreover, we again use that $G$ and the $S_k$ are deterministic functions.
\begin{alignat*}{3}
    &\mathrlap{\Prob(Z^* = z^*\mid Z=z, R = r, X_R=x_r, Y_R(\cdot) = y_r(\cdot), S(Z^*)=S(Z), G(Z^*)=G(Z))}\quad &&&&\\[1ex]
    && &\mathrlap{\propto \one\{G(z^*)=\,G(z)\}\, \cdot}\qquad && \tag{by (D) }
      \\[-0.5ex]
    &&&& \prod_{k=1}^K\,\,
    &\begin{aligned}[t]
        \Prob(R_k = r_k \mid\, &R_{[k-1]} = r_{[k-1]}, X_{R_{[k-1]}}=x_{r_{[k-1]}},Y_{R_{[k-1]}}(\cdot) = y_{r_{[k-1]}}(\cdot),\\ &Z_{[k-1]} = z_{[k-1]}, S_{k-1}(z^*_{[k-1]}) =
    S_{k-1}(z_{[k-1]}))\,\cdot
    \\[-1.5ex]
    \end{aligned} \tag{by \eqref{eq:zstar-sel}}
    \\[1.0ex]
    &&&& &\begin{aligned}[t]
        \Prob(X_{R_k} = x_{r_k}, Y_{R_k}(\cdot) = y_{r_k}(\cdot) \mid\,
    &R_{[k]} = r_{[k]}, X_{R_{[k-1]}}=x_{r_{[k-1]}},Y_{R_{[k-1]}}(\cdot) = y_{r_{[k-1]}}(\cdot),\\ &Z_{[k-1]} = z_{[k-1]}, S_{k-1}(z^*_{[k-1]}) =
    S_{k-1}(z_{[k-1]}))\,
    \cdot
    \end{aligned} \tag{by \eqref{eq:zstar-sel}}
    \\[1.0ex]
    &&&& &\Prob(Z_k = z_k \mid R_{[k]} = r_{[k]}, X_{R_{[k]}} =
    x_{r_{[k]}}, Y_{R_{[k-1]}}(z_{[k-1]}) = y_{r_{[k-1]}}, Z_{[k-1]} =
    z_{[k-1]})\, \cdot \tag{by \eqref{eq:zstar-indep}}
    \\[1ex]
    &&&& &\Prob(Z^*_k = z^*_k \mid R_{[k]} = r_{[k]}, X_{R_{[k]}} =
    x_{r_{[k]}}, Y_{R_{[k-1]}}(z^*_{[k-1]}) = y^*_{r_{[k-1]}}, Z^*_{[k-1]} =
    z^*_{[k-1]})\, \cdot \tag{by \eqref{eq:zstar-indep}}
    \\[1ex]
    &&&& &\one{\{S_k(z^*_{[k]})=S_k(z_{[k]})\}} \tag{by (D)}
    \\[1ex]
    && &\mathrlap{\propto \one\{G(z^*)=\,G(z), S(z^*)=S(z)\}\,\cdot q(z^*\mid r,x_r,y^*_r).}\qquad && \tag{by \eqref{eq:zstar-id}}
\end{alignat*}
}
}


The result now follows from expanding the formula of the selective randomization p-value and applying the factorization above:
{\small
    \begin{align*}
        P_\textnormal{sel} &= \sum\limits_{z^* \in \calZ} \one\{T(z^*,R,X_R,Y_R(\cdot)) \leq T(Z,R,X_R,Y_R(\cdot))\}\cdot
        \Prob(Z^*=z^*\mid R,X_R,Y_R(\cdot),\, S(z^*)=S(Z), G(z^*)=G(Z))\\
        &=\frac{\sum\limits_{z^* \in \calZ}\, \one\big\{T(z^*,R,X_R,Y_R(\cdot)) \leq T(Z,R,X_R,Y_R(\cdot))\big\}\cdot\one\big\{ G(z^*)=G(Z), S(z^*)=S(Z)\big\}\cdot q(z^*\mid R,X_R,Y_R(\cdot))}{\sum\limits_{z^* \in \calZ}\,\one{\big\{G(z^*)=G(Z), S(z^*)=S(Z)\big\}}\cdot q(z^*\mid R,X_R,Y_R(\cdot))}.
    \end{align*}
}This expression can be computed because $T$, $S$ and the randomization distribution of $Z$ are assumed to be imputable with respect to $G$. 


\section{Computation}\label{sec:computation}



Recall that the feasible treatment assignment space/ the support of the distribution of $Z$ is denoted as $\treatmentSpace = \bigtimes_k \treatmentSpace_k$. In order to shorten the notation, we use $\treatmentSpace_{s,g} := \{z \in \treatmentSpace: S(z) = s, G(z) = g\}$ to denote the support of $Z$ conditional on $S(Z)=s$ and $G(Z)=g$, where $s \in \calS$ and $g \in \calG$. (To simplify the notation, we omit the dependence on $R, X_R$ and $Y_R(\cdot)$.)
First, we define the Monte Carlo approximation of the selective randomization p-value and discuss the involved computational challenges stemming from conditioning on $S$ and $G$. Then, we propose two sampling schemes -- rejection sampling and the random walk Metropolis-Hastings sampler -- to approximate the selective randomization p-value. Lastly, we address computing confidence intervals.


\subsection{Monte Carlo Approximation}\label{sec:compute.challenge}

For tests based on randomization or permutation, it is typically infeasible to enumerate all possible treatment configurations as their number grows exponentially in the number of recruited units. Therefore, we use a Monte Carlo approximation of the p-value instead: We generate $M$ samples $Z^{(t)}$, $t \in \{1,\ldots,M\}$, from the selective randomization distribution 
and compute
\begin{align*}
    \hat{P}_M:= \frac{1+\sum_{t=1}^M\, \mathbf{1}\{T(Z^{(t)},R,X_R,Y_R(\cdot)) \leq T(Z,R,X_R,Y_R(\cdot))\}}{1+M}.
\end{align*}
Here the additional one is added to the denominator and numerator to ensure that $\hat{P}_M$ is a valid p-value. We remark that the additional one can also be interpreted as including $Z$ as an additional sample which yields the term $\mathbf{1}\{T(Z,R,X_R,Y_R(\cdot)) \leq T(Z,R,X_R,Y_R(\cdot))\}=1$.

Generating feasible samples $Z^{(t)}$, that is $S(Z^{(t)}) = S(Z)$ and $G(Z^{(t)}) = G(Z)$, 
poses additional challenges compared to traditional sampling problems. 
First, the probability of the conditioning event might be exceedingly small. For example, when one conditions on a continuous selection statistic attaining a specific value, or when the potential outcomes are imputed under some pre-defined null that is at odds with the observed data, it becomes difficult to generate samples that yield the observed selection, cf.\ \Cref{sec:confint}.
Second, the feasible treatment space~$\treatmentSpace$ typically has a simple topology and is convenient to sample from, e.g.\ for completely randomized or Bernoulli trials. 
However, the conditioning restricts the support to $\treatmentSpace_{S(Z),G(Z)}$ and
usually disrupts this helpful structure.
Third, since $\treatment$ is a discrete random vector, gradient-based sampling methods, such as Hamiltonian Monte Carlo or the Metropolis-adjusted Langevin algorithm, are not suitable.

In the following, we introduce two Monte Carlo sampling schemes: (1) rejection sampling and (2) the random walk Metropolis-Hastings (RWM) sampler from the class of Markov Chain Monte Carlo (MCMC) algorithms. While the former is easy to implement and does not require hyper-parameter tuning, it may be computationally expensive. The latter generates samples more efficiently but requires careful choice of hyper-parameters and access to the randomization distribution of $Z$ conditional on some entries of $Z$. Therefore, we recommend rejection sampling when the success probability of generating a feasible sample is not exceptionally low and the RWM sampler otherwise.

\subsection{Rejection Sampling}
\label{sec:rejection.sampling}

Rejection sampling is a versatile procedure that allows for arbitrary treatment assignment schemes and conditioning events. 
In every step of the algorithm, we generate a sample $Z^*$ from the non-selective randomization distribution $\Prob(Z=\cdot \mid R,X_R,Y_R(\cdot))$. If $Z^*$ is feasible, i.e.\ $Z^* \in \calZ_{S(Z),G(Z)}$, we accept it; otherwise, $Z^*$ is rejected. We repeat this process until the desired number of $M$ samples is reached. The pseudocode of this process is summarized in \Cref{algo:rejection.sampling}.

Rejection sampling is a simple and flexible algorithm, but it may entail considerable computational costs. Suppose the success probability $\Prob(\treatment^* \in \treatmentSpace_{S(Z),G(Z)} \mid R,X_R,Y_R(\cdot),Z)$ 
equals $p_0$; then the expected number of draws required to obtain one acceptable treatment is $1/p_0$. When $p_0$ is extremely small (potentially scaling exponentially with the number of randomized units), the expected number of candidates that need to be generated to produce one accepted sample can be considerably high. We encountered this computational challenge in the simulations in \Cref{sec:simulation-study} as well: In the trials where only one subgroup is selected, the conditional randomization test with rejection sampling takes around 150 times longer than the Monte Carlo approximation of the randomization test without conditioning; in other words, roughly only one out of 150 proposals is accepted.

\begin{algorithm}[t]
\caption{Rejection sampling}\label{algo:rejection.sampling}
\begin{algorithmic}
\STATE \textbf{Input}: data $\treatment$, $R,X_R,Y_R(\cdot)$ (observed or imputed); selection and test statistic functions $S(\cdot)$, $G(\cdot)$ and $T(\cdot)$; treatment assignment mechanism $\Prob(Z=\cdot \mid R,X_R,Y_R(\cdot))$; number of samples $M$.
\STATE \textbf{Initialization}: the current number of accepted treatments $t = 0$.
\WHILE{$t < M$}
\STATE Sample $\treatment^*$ from the distribution $\Prob(Z=\cdot \mid R,X_R,Y_R(\cdot))$
\IF{$\treatment^* \in \treatmentSpace_{S(\treatment),G(\treatment)}$}
\STATE Accept $\treatment^*$ and let $\treatment^{(t)} \leftarrow \treatment^*$\\ $t \leftarrow t + 1$
\ELSE
\STATE Reject $\treatment^*$ and continue.
\ENDIF
\ENDWHILE
\STATE Compute $\hat{P}_M = \Big({1+\sum_{t=1}^{M} \one\{T(\treatment^{(t)},R,X_R,Y_R(\cdot)) \leq T(\treatment, R,X_R,Y_R(\cdot))\}}\Big)/({1+M})$.
\STATE \textbf{Output}: $\hat{P}_M$.
\end{algorithmic}
\end{algorithm}

\subsection{Random Walk Metropolis-Hastings Algorithm}\label{sec:RWM}

Markov Chain Monte Carlo (MCMC) algorithms are a class of sampling schemes which construct a Markov chain that has the target distribution -- in our case the selective randomization distribution -- as its stationary distribution. They can be visualized as a random walk in the state space $\calZ$ that yields one sample at every step. There are various MCMC algorithms but we focus on the popular random walk Metropolis-Hastings (RWM) sampler 
\parencite{metropolis1953equation} in this work.



\subsubsection{Algorithm}
We choose a probability measure $\rho$ on the power set of $R$ which guides the ``step length'' of the random walk and initialize the algorithm at the observed treatment assignment $\treatment$. Then we repeat the following iteration $M$ times. In the $t$-th step, the previous sample (state of the Markov chain) is denoted by $Z^{(t-1)}$; to obtain a new sample, we randomly select a subset $A \subseteq R$ of entries of $Z^{(t-1)}$ following $\rho$ for updating their values. Given $A$, we fix the treatment assignments of units not in $A$ and generate new treatment assignments for the units in $A$. That is, we sample a proposal $Z^*$ from the conditional distribution $\Prob(Z=\cdot \mid R,X_R,Y_R(\cdot), Z_{A^C} = Z^{(t-1)}_{A^C})$. Upon obtaining $\treatment^*$, we set $\treatment^{(t)} = \treatment^*$ if $\treatment^* \in \treatmentSpace_{S(Z),G(Z)}$ and reject $\treatment^*$ and use $\treatment^{(t)} = \treatment^{(t-1)}$, otherwise.
The details are summarized in \Cref{algo:MCMC}.
\begin{algorithm}[t]
\caption{Random walk Metropolis-Hastings algorithm}\label{algo:MCMC}
\begin{algorithmic}
\STATE \textbf{Input}:
data $\treatment$, $R,X_R,Y_R(\cdot)$ (observed or imputed); selection and test statistic functions $S(\cdot)$, $G(\cdot)$ and $T(\cdot)$; treatment assignment mechanism $\Prob(Z=\cdot \mid R,X_R,Y_R(\cdot))$; number of samples $M$; a probability measure $\rho$ on the power set of $R$; the burn-in size $b$.
\STATE \textbf{Initialization}: Set $\treatment^{(0)} = \treatment$ (observed value).
\FOR{t $= 1: M$}
\STATE Sample a subset of units $A$ from $\rho$\\
Sample $\treatment^*$ from $\Prob(Z=\cdot \mid R,X_R,Y_R(\cdot), Z_{A^C} = Z^{(t-1)}_{A^C})$ 
\IF{$\treatment^* \in \treatmentSpace_{S(\treatment),G(\treatment)}$}
\STATE $\treatment^{(t)} \leftarrow \treatment^*$.
\ELSE
\STATE $\treatment^{(t)} \leftarrow \treatment^{(t-1)}$.
\ENDIF
\ENDFOR
\STATE Compute $\hat{P}_M = \left({1+\sum_{t=b+1}^{M} \one\{T(\treatment^{(t)}, R,X_R,Y_R(\cdot)) \leq T(\treatment, R,X_R,Y_R(\cdot))\}}\right)/({1+(M-b)})$.
\STATE \textbf{Output}: $\hat{P}_M$.
\end{algorithmic}
\end{algorithm}

This procedure is quite general and abstract but significantly simplifies for common treatment assignment distributions such as the Bernoulli trial or a completely randomized design.

\begin{example}[Bernoulli trial]\label{exam:Bernoulli.trial}
Consider the simple Bernoulli trial described in \Cref{sec:notation}. For each stage $k$, we can specify a window size/step length $h_k \in \{1,\ldots,\lvert R_k \rvert\}$ and randomly choose $h_k$ entries from $\treatment_k^{(t)}$ to replace with independent draws from a $\mathrm{Bernoulli}(p)$-distribution. (This corresponds to using a product measure $\rho = \bigtimes_k \rho_k$, where each $\rho_k$ is supported on the subsets of $R_k$ with size $h_k$ and assigns equal probability to them.)
\end{example}


\begin{example}[Completely randomized design]\label{exam:CRD}
This treatment assignment scheme randomly assigns \emph{prefixed} numbers of units to the different arms of the study with uniform probability. Suppose that at each stage a completely randomized design is conducted independently, and let $\treatment_k^{(t)}$ be the treatment assignment of stage $k$ in $\treatment^{(t)}$. Similarly to the Bernoulli trial, we can choose a window-size $h_k \in \{2,\ldots,\lvert R_k \rvert\}$ for each stage $k$ and randomly select $h_k$ entries from $\treatment_k^{(t)}$. Randomly shuffling these $h_k$ entries generates a new proposal $\treatment_k^*$ and simultaneously preserves the condition for fixed numbers of participants in every arm. This last step of shuffling corresponds to sampling from the conditional distribution $\Prob(Z=\cdot \mid R,X_R,Y_R(\cdot), Z_{A^C} = Z^{(t-1)}_{A^C})$.
\end{example}

Markov Chain Monte Carlo algorithms typically involve hyper-parameters, such as the window size(s) in \Cref{exam:Bernoulli.trial,exam:CRD}, which need to be tuned to get a fast convergence to the stationary distribution. In this article, we choose them so that they maximize the Mean Squared Euclidean Jump Distance (MSEJD) defined as
\begin{align*}
    S^2 = \frac{1}{M-1} \sum_{t=1}^{M-1} \|\treatment^{(t+1)} - \treatment^{(t)}\|_2^2.
\end{align*}
This (empirical) efficiency metric is equivalent to a weighted sum of the lag-$1$ auto-correlation for stationary chains with finite posterior variance \parencite{sherlock2010random} and a higher MSEJD value indicates greater efficiency.
We remark that for binary treatments, $S^2$ is equivalent to the average Hamming distance between adjacent samples $\treatment^{(t)}$ and $\treatment^{(t+1)}$, $t \in [M-1]$.
Compared to other metrics for choosing hyper-parameters, e.g.\ the rejection probability of a new proposal, the significance of MSEJD lies in its ability to not only account for the frequency of acceptance but also capture the magnitude of the differences $\treatment^* - \treatment^{(t)}$ among accepted proposals $\treatment^*$.


\subsubsection{Theory}
As explained in the previous subsection, we need to specify the step size (window size) of the RWM sampler via the probability measure $\rho$. This hyper-parameter affects the convergence speed of the Markov chain but may additionally impact the limiting distribution it converges towards. If the support of the selective randomization distribution consists of disconnected sets and a small step size is chosen, the random walk may not be able to reach all such sets and consequently does not converge to the target distribution $\Prob(Z = \cdot \mid R,X_R,Y_R(\cdot), S(Z), G(Z))$. Yet, even in this case, the Monte Carlo approximation converges to a valid selective p-value.

To show this result, we first introduce the concept of \emph{communication}: For a Markov chain with state space $\calZ$ we say that $z\in \calZ$ and $z' \in \calZ$ communicate if and only if $z$ can be reached from $z'$ in finitely many steps of the Markov chain with positive probability and vice versa. We denote the resulting equivalence relation $\equiv_C$ and define the communication class of a state $z \in \calZ$ as follows $C_{s,g}(z) := \{z' \in \calZ_{s,g}\colon z \equiv_C z'\}$, where $s \in \calS$ and $g \in \calG$.
Analogously to Definition~\ref{def:rand-p-value}, we can define the selective randomization p-value that additionally conditions on the communication class of the starting point of a Markov chain.

\begin{definition}\label{def:com-class-p-value}
    In the setting of \Cref{def:rand-p-value}, let $Z^*$ have the same distribution as $Z$ and fulfill $Z^*\indep Z\,\vert\,R,X_R,Y_R(\cdot)$. Then, the selective randomization p-value conditional on the communication class of $Z$ is given by
    \begin{multline*}
\tilde{P}_{\mathrm{sel}}:=\Prob\big(T(Z^*,R,X_R,Y_R(\cdot)) \leq T(Z,R,X_R,Y_R(\cdot)) \,\,\big\vert\,\, R,X_R,Y_R(\cdot), Z, S(Z^*) = S(Z), G(Z^*) = G(Z),\\[-2ex]
Z^* \in C_{S(Z),G(Z)}(Z)\big).
    \end{multline*}
\end{definition}
\begin{remark}
    Conditioning on $S(Z^*) = S(Z)$ and $G(Z^*) = G(Z)$ in the equation above is not necessary because these conditions follow from $Z^* \in C_{S(Z),G(Z)}$. Nonetheless, we keep $S$ and $G$ to emphasize that $\tilde{P}_{\mathrm{sel}}$ involves additionally conditioning.
\end{remark}

This modified selective randomization p-value indeed preserves the selective type-I error control.
\begin{proposition}
    Let $\alpha \in [0,1]$. The selective randomization p-value as given in Definition \ref{def:com-class-p-value} stochastically dominates the uniform distribution on $[0,1]$:
    \begin{equation*}
        \Prob(\,\tilde{P}_{\mathrm{sel}} \leq \alpha \,\vert\, R,X_R,Y_R(\cdot), S(Z)=s, G(Z)=g, Z \in C_{s,g}(z)) \leq \alpha\quad \forall\, s\in \calS, g\in \calG, z \in \calZ.
    \end{equation*}
    Consequently, a test which rejects the selected null hypothesis when $\tilde{P}_{\text{sel}} \leq \alpha$ controls the selective type-I error, i.e.
    \begin{equation*}
        \Prob(\,\tilde{P}_{\text{sel}} \leq \alpha \mid R,X_R,Y_R(\cdot), S(Z)=s) \leq \alpha\quad \forall\, s\in \calS.
    \end{equation*}
\end{proposition}
\begin{proof}
    See \textcite[Thm.\ 1]{zhang_2023_randomization_test}, where the partition of the treatment assignment space $\calZ$ is given by the preimages of $(S,G)$ and the equivalence relation $\equiv_C$, cf.\ \textcite[Prop.\ 1]{zhang_2023_randomization_test}.
\end{proof}
Having established the validity of the modified selective randomization p-value, we now show that the Monte Carlo approximation based on RWM samples indeed converges to $\tilde{P}_{\mathrm{sel}}$.
\begin{proposition}\label{prop:reversible}
    The Markov chain defined in Algorithm~\ref{algo:MCMC} and restricted to the state space $C_{S(Z),G(Z)}(Z)$ is aperiodic, irreducible and reversible. Its stationary distribution is $\Prob(Z^* = \cdot \mid R,X_R,Y_R(\cdot), Z^* \in C_{S(Z),G(Z)}(Z))$, where $Z^*$ has the same distribution as $Z$ and $Z^* \indep Z \mid R,X_R,Y_R(\cdot)$.
    In addition, there exists {$\sigma^2(\treatment, R,X_R,Y_R(\cdot))~<~\infty$} such that
    \begin{equation*}
        \sqrt{M}(\hat{P}_M - \tilde{P}_\mathrm{sel})
    \mid \treatment , R,X_R,Y_R(\cdot) \,\stackrel{d}{\to}\,
        \calN\left(0, \sigma^2(\treatment, R,X_R,Y_R(\cdot))\right), \quad \text{as } M \to \infty.
    \end{equation*}
\end{proposition}


\begin{proof}\label{proof:prop:reversible}
For any $z$ in the state space, the probability of staying at itself is positive, and thus the Markov Chain is
aperiodic. Moreover, it is irreducible since the state space forms one communication class.

To ease readability, we use the notation
$\pi(z) := \Prob(Z^* = z \mid R,X_R,Y_R(\cdot), Z^* \in C)$, where $z \in \calZ$ and $C:=C_{S(Z),G(Z)}(Z)$, in the following; furthermore, $P(z,z')$ denotes the transition probability of moving from state $z$ to~$z'$. The proposal distribution has the probability mass function
\begin{equation*}
    f(z'\mid z) = \sum_{A \subseteq R} \Prob(Z^* = z' \mid R,X_R,Y_R(\cdot), Z^*_{A^C} = z_{A^C})\, \rho(A)\qquad \text{for } z \in C, z' \in \calZ.
\end{equation*}
Hence, the transition probability from $z$ to $z'$ is given by
\begin{equation*}
    P(z,z') = f(z'\mid z)\, \one\{z' \in C\}\, +\, \one\{z=z'\}\!\sum_{z'' \in\calZ}
    \one\{z'' \notin C\}\, f(z''\mid z),
\end{equation*}
where the second term captures the probability of $f(\cdot \mid z)$ proposing a state outside of $C$.

If the Markov chain fulfills the detailed balance equation
\begin{equation}\label{eq:detailed-balance}
    \pi(z)\, P(z,z') = \pi(z') \, P(z',z)\qquad
    \text{for all }z,z' \in C,
\end{equation}
then it is reversible and $\pi$ is its (unique) stationary distribution \parencite[Cor.\ 1.17, Prop.\ 1.20]{levin2017markov}. For $z = z'$, \eqref{eq:detailed-balance} is trivially true; hence, we can assume $z \neq z'$ without loss of generality. We obtain the detailed balance equation as follows
\begin{align*}
    \frac{P(z,z')}{P(z',z)} &=
    \frac{\sum_{A \subseteq R} \Prob(Z^* = z' \mid R,X_R,Y_R(\cdot), Z^*_{A^C} = z_{A^C})\, \rho(A)\,\one\{z' \in C\}}{\sum_{A \subseteq R} \Prob(Z^* = z \mid R,X_R,Y_R(\cdot), Z^*_{A^C} = z_{A^C}')\, \rho(A)\,\one\{z \in C\}}\\[1.5ex]
    &= \frac{\sum_{A \subseteq R}
    \Prob(Z^*=z'\mid R,X_R,Y_R(\cdot))\,/\,\Prob(Z^*_{A^C}=z_{A^C}\mid R,X_R,Y_R(\cdot))\, \one\{z'_{A^C}=z_{A^C}\}\,
    \rho(A)\,\one\{z' \in C\}}{\sum_{A \subseteq R}
    \Prob(Z^*=z \mid R,X_R,Y_R(\cdot))\,/\,\Prob(Z^*_{A^C}=z'_{A^C}\mid R,X_R,Y_R(\cdot))\, \one\{z_{A^C}=z'_{A^C}\}\,
    \rho(A)\,\one\{z \in C\}}\\[1.5ex]
    &= \frac{\Prob(Z^*=z'\mid R,X_R,Y_R(\cdot))\,\one\{z'\in C\}}{\Prob(Z^*=z\mid R,X_R,Y_R(\cdot))\,\one\{z\in C\}} \cdot \frac{\sum_{A \subseteq R}
    1/\Prob(Z^*_{A^C}=z_{A^C}\mid R,X_R,Y_R(\cdot))\, \one\{z'_{A^C}=z_{A^C}\}\,
    \rho(A)}{\sum_{A \subseteq R}
    1/\Prob(Z^*_{A^C}=z'_{A^C}\mid R,X_R,Y_R(\cdot))\, \one\{z_{A^C}=z'_{A^C}\}\,
    \rho(A)}\\[1.5ex]
    &= \frac{\Prob(Z^*=z'\mid R,X_R,Y_R(\cdot))\,\one\{z'\in C\}}{\Prob(Z^*=z\mid R,X_R,Y_R(\cdot))\,\one\{z\in C\}} \cdot \frac{\Prob(Z^* \in C \mid R,X_R,Y_R(\cdot))}{\Prob(Z^* \in C \mid R,X_R,Y_R(\cdot))} \cdot 1 = \frac{\pi(z')}{\pi(z)}.
\end{align*}

Lastly, since the Markov Chain is aperiodic and irreducible, it converges to its stationary distribution $\pi$ and is geometrically ergodic, see \textcite[Theorem 4.9]{levin2017markov}. Moreover, any geometrically ergodic reversible Markov chain admits a central limit theorem for all functions with a finite second moment with respect to the stationary distribution \parencite{gilks1996markov}. Since the indicator function $\one{\{T(\treatmentElem^*, R,X_R,Y_R(\cdot)) \leq T(\treatment, R,X_R,Y_R(\cdot))\}}$ is bounded, we obtain the desired convergence result for the Monte Carlo approximation of $\tilde{P}_{\mathrm{sel}}$.
\end{proof}

\subsection{Constructing Confidence Intervals}\label{sec:compute.confidence.interval}

Here, we discuss the computational details of constructing confidence sets via inversion of tests as outlined in \eqref{eq:confidence-set}.
Let $\calI$ denote a user-specified interval of potential constant treatment effects and suppose $\calI$ is of finite length.
For a (possibly non-monotone) selective randomization p-value, one can discretize $\calI$ into bins of width $\varepsilon > 0$, and evaluate the p-value at the representative value of each bin.
The confidence set shall comprise all bins associated with p-values above the significant level.
The testing of different representative values of $\tau$ can be computed independently, allowing the entire procedure to be accelerated through parallel computing.
If the p-values are monotone, the Robbins-Monro algorithm and the bisection method can be employed to establish the lower and upper limits of this interval in a more efficient way \parencite{garthwaite1996confidence, wang_randomization_2020}. Moreover, we can reduce the number of Monte Carlo samples by using sequential p-values \parencite{besag_sequential_1991,fay_using_2007,silva_optimal_2013,fischer_sequential_2024,howes_markov_2024}. These methods allow to reduce the number of samples if it becomes clear early on that the estimate of the p-value is far away from the levels of interest, that is $\alpha/2$ and $1-\alpha/2$ for confidence intervals and $1/2$ for estimation, and hence does not require a high resolution.

\section{Details of the Simulation Study on Hold-out Units}\label{app:details-hold-out}
We use the data-generating mechanism of the two-stage enrichment trial as described in Section~\ref{sec:simulation-study} and modify it as follows.

First, we consider the common scenario that there is treatment effect heterogeneity between groups, cf.\ \textcite{marston_predicting_2020}. To this end, we generate the potential outcomes as $Y_i(1) = Y_i(0)+1$ for participants with high genetic risk and as $Y_i(1) = Y_i(0)$ for participants with low genetic risk; $Y_i(0)$ are i.i.d.\ draws from a standard normal distribution. Hence, the treatment effect in the former group is $1$ and $0$ in the latter.

Second, we lower the number of recruited patients to 16 in the first and second stage each. This allows us to precisely compute the p-values for different treatment effects $\tau$ without the need for a Monte Carlo approximation. Therefore, we can exclude numerical errors as a cause of our findings.

Third, we employ two different selection rules. For the study with hold-out units, we decide after the first stage for which group(s) we test the (partially) sharp null hypothesis. Hence, we use $S$ as defined in~\eqref{eq:selection-rule} with data from the first stage. Thus, the selection depends on $Z_1$ but not $Z_2$. For the study without hold-out units, we choose the null hypothesis only after completing the second stage. Hence, we compute $S$ according to~\eqref{eq:selection-rule} with data from both stages. Consequently, $S$ depends on $Z_1$ and $Z_2$ in this scenario.

We simulate 10 different datasets. Since the true treatment effect in the group with high genetic risk is larger, we observe that the selection rules choose this group in the majority of cases. Since we employ two different rules, they may disagree, however. In order to make a fair comparison between the two adaptive designs, we only consider scenarios where both rules choose to test the null hypothesis for the group with high genetic risk. Among the remaining datasets, we compute the p-value curve via discretizing the interval $[-2, 3]$. Figure~\ref{fig:hold-out} depicts the most jagged p-value curve we found. The source code of this simulation study is available on \url{https://github.com/tobias-freidling/selective-randomization-inference}.

\section{Additional Results on the Numerical Experiments}\label{app:additional-results-num-exp}

\color{black}
\subsection{Test Statistic}\label{app:test.statistic}

    For the choice of test statistic, standard estimators of average treatment effect, such as the inverse propensity weighting estimator or its normalized variant, provide decent power.
    Furthermore, we emphasize that the test statistic can depend on
    the intermediate selection. For example, in the running example of
    a two-stage experiment, the test statistic can be constructed to depend on the selection made after the first stage.
    Particularly, we suggest up-weighting the units in the second stage when the first-stage selection is extreme.
    The rationale is that if the first-stage selection is extreme, a substantial amount of information has already been consumed by the selection. Consequently, the information remaining for inference, following the data-carving perspective \parencite{fithian_optimal_2017}, is reduced and should therefore be down-weighted.

    We illustrate selection-dependent test statistics below using a simulated example.
    We adopt the simulation setup described in the simulation section of the manuscript. Recall that the setup involves two subgroups ($X_i \in \{\text{low}, \text{high}\}$) and a two-stage recruitment procedure based on the first-stage normalized difference in ATE between the two groups.
    The second-stage recruitment rule is as follows: if one subgroup's ATE is substantially larger than the other, only that subgroup is recruited; otherwise, both subgroups are recruited.
    We compare several test statistics that interpolate between first-stage-only and second-stage-only tests.
    Let $w \in [0, 1]$ be a weight parameter, and let $T_1$ and $T_2$ denote the standardized difference in means estimator of ATE from the first and second stages, respectively. Here the estimator is computed based on the subgroup selected in the first stage.
     \begin{figure}[h!]
  \centering
  \includegraphics[scale=0.8]{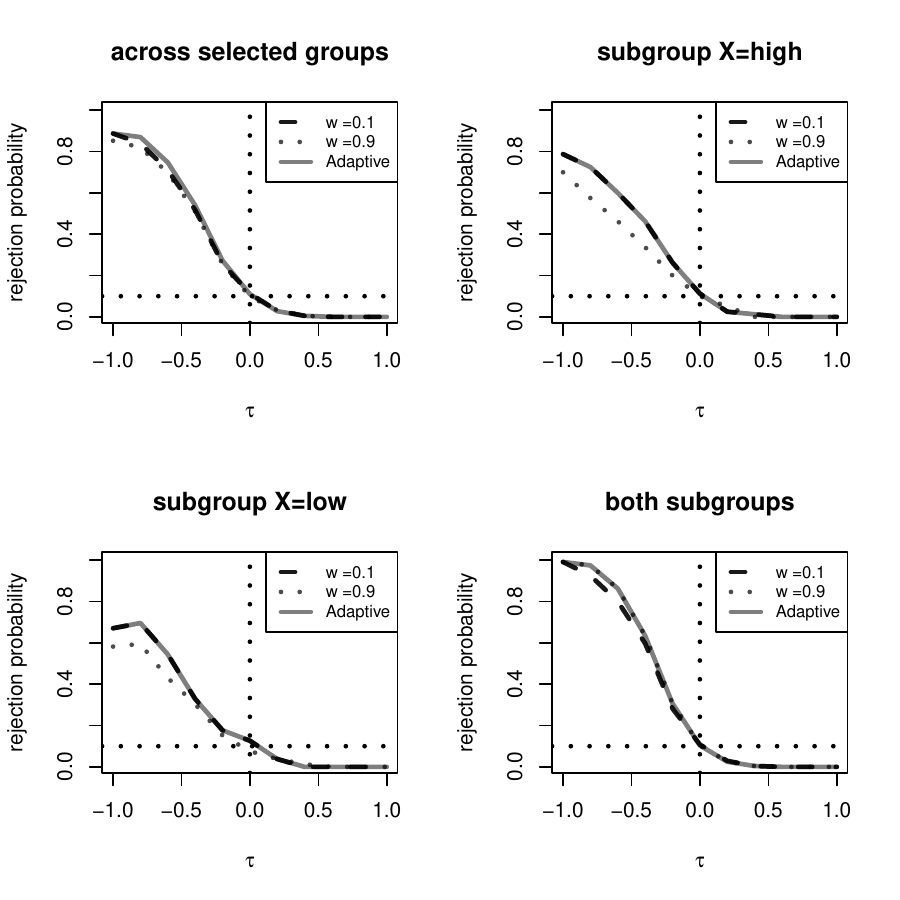}
  \caption{Power comparison of various test statistics. All test statistics control the conditional type I error at the target level ($0.1$) when $\tau = 0$.
  $T_{0.1}$ (dashed) is more powerful for extreme selections (subgroup $X=$ high or low), while $T_{0.9}$ (dotted) performs better when both subgroups are selected.
  The adaptive statistic $T_{\mathrm{ada}}$ combines the strengths of both and achieves the highest aggregated power across selections.
  }
  \label{fig:default_weight_power}
\end{figure}
    We consider three test statistics. $T_{0.1}$ and $T_{0.9}$, which are defined as
    \begin{align*}
        T_{0.1} &= \sqrt{0.1}\,T_1 + \sqrt{0.9}\,T_2,\qquad \text{where } w=0.1,\\
        T_{0.9} &= \sqrt{0.9}\,T_1 + \sqrt{0.1}\,T_2,\qquad \text{where } w=0.9,
    \end{align*}
    are selection-agnostic and respectively up- and down-weight the units in the second stage. Furthermore, we consider the adaptive statistic
    \begin{equation*}
        T_{\mathrm{ada}} =
        \begin{cases}
            \sqrt{0.1}\,T_1 + \sqrt{0.9}\,T_2, & \text{if one subgroup is selected}, \\
            \sqrt{0.9}\,T_1 + \sqrt{0.1}\,T_2, & \text{if both subgroups are selected}.
        \end{cases}
    \end{equation*}
    This statistic depends on the selection: if the first-stage selection is extreme, i.e.,  only one subgroup is selected, a certain amount of information has already been ``used'' in the selection step, so the second stage should be given greater weight in the inference.




    We provide the power plots in \Cref{fig:default_weight_power}. Regardless of the choice of test statistic, all methods control the type I error after selection, as the rejection probability remains at the target level $0.1$ when $\tau = 0$.
    In terms of power, the statistic $T_{0.1}$ (dashed), which places more weight on the second-stage data, is more powerful than $T_{0.9}$ (dotted) when the selection is extreme (e.g., subgroup $X=$ high or subgroup $X=$ low). In contrast, $T_{0.9}$, which emphasizes the first-stage data of a larger sample size, outperforms $T_{0.1}$ when both subgroups are selected (mild selection).
    The adaptive statistic $T_{\mathrm{ada}}$ chooses the more powerful statistic between $T_{0.1}$ and $T_{0.9}$ conditional on the selection event, yielding higher aggregated power across selections.
    We comment that, if the gap between $T_{0.9}$ and $T_{0.1}$ widens in the ``both subgroups'' scenario, the adaptive test statistic can achieve even greater overall power compared with $T_{0.9}$ and $T_{0.1}$.

\color{black}
\subsection{Window Size of the RWM Sampler}\label{appe:sec:figures}

\begin{figure}[h]
        \centering
        \begin{minipage}{0.96\textwidth}
                \centering
                \includegraphics[clip, trim = 0.5cm 0cm 0cm 0cm, width = \textwidth]{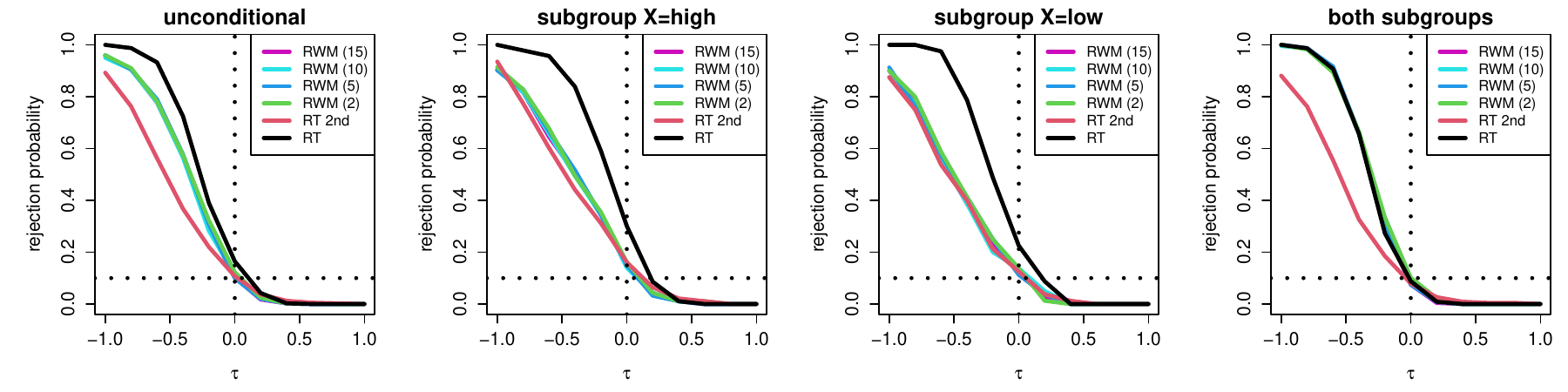}
                \subcaption[(a)]{Rejection probability.}
        \end{minipage}
        \begin{minipage}{0.96\textwidth}
                \centering
                \includegraphics[clip, trim = 0.5cm 0cm 0cm 0cm, width = \textwidth]{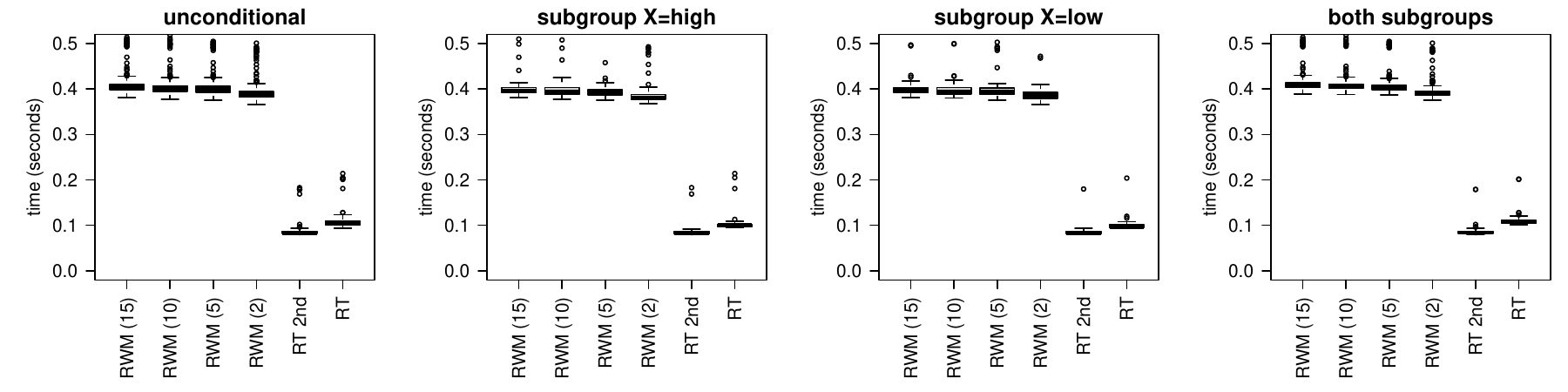}
                \subcaption[(b)]{Computation time.}
        \end{minipage}
    \caption{Comparison of RWM sampler with different window sizes regarding rejection probability and computation time. We include RT and RT 2nd for reference and exclude rejection sampling due to its slow computation speed.}
    \label{fig:window.size}
\end{figure}

We consider the simulation study of Section~\ref{sec:simulation-study} and vary the window size of the RWM sampler taking values in $\{2, 5, 10, 15\}$. We plot the rejection probability for different null hypotheses as well as the computation time in \Cref{fig:window.size}.

The power comparison of SRT, RT 2nd, and RT is similar to that in \Cref{sec:simulation-study}. Moreover, different window sizes yield almost identical approximations. The running time of the RWM sampler largely remains the same (or increases only mildly) as the window size increases. This is because -- regardless of whether the random walk transitions to the new proposal or not -- every iteration of the algorithm generates a new sample, in contrast to the rejection sampling algorithm.
The main difference in computation time across window sizes is just a
consequence that permuting $2$ elements is less expensive than
permuting $10$.
The RWM sampler is slower than RT 2nd and RT due to the extra step verifying whether the proposal is acceptable.

\subsection{Confidence Intervals}\label{app:confidence-intervals}
\begin{table}[htb]
\centering
\caption{
Empirical coverage of one-sided confidence intervals based on different randomization tests (significance level $\alpha=0.1$).
}\label{tab:coverage}
\begin{tabular}{lc|cccc}
\toprule
 & Frequency &  SRT-RWM & SRT-RS   & RT 2nd & RT \\
\midrule
 High risk group selected & 20\% & 0.861 & 0.873 & 0.873 & 0.722\\
 Low risk group selected & 20\% & 0.921 & 0.934 & 0.949 & 0.795\\
 Both groups selected & 60\% & 0.909 & 0.918  & 0.909  & 0.922 \\
 \midrule
 Average for the selected group(s) & & 0.910  & 0.920  & 0.910 & 0.858 \\
\bottomrule
\end{tabular}
\end{table}

\begin{figure}[htb]
    \centering

                \includegraphics[clip, trim = 0.5cm 0cm 0cm 0cm, width = \textwidth]{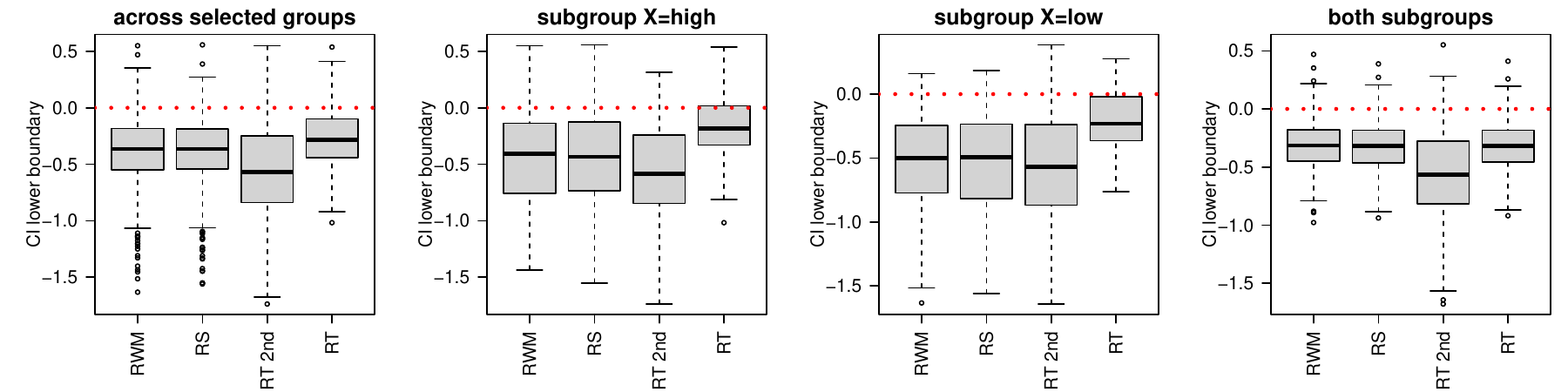}
    \caption{Distribution of the lower boundary of one-sided confidence intervals obtained from different randomization tests.
    }
    \label{fig:default.confidence.interval}
\end{figure}

In the context of the simulation study of Section~\ref{sec:simulation-study}, we construct one-sided confidence intervals for the homogeneous treatment effect $\tau$. To this end, we specify an interval of treatment effects, discretize it and conduct a hypothesis test for each resulting bin. Simulating multiple datasets, we can compare the empirical coverage and length of the obtained confidence intervals.

In Table~\ref{tab:coverage}, we compare the naive randomization test (RT), the second stage randomization test (RT 2nd) as well as our proposed selective randomization test (SRT) computed either via rejection sampling (RS) or the random walk Metropolis-Hastings (RWM) sampler. We notice that the empirical coverage of RT 2nd and SRT is close to the nominal value of 90\% whereas RT yields lower empirical coverage, especially when only one subgroup is selected.

In Figure~\ref{fig:default.confidence.interval}, we compare the distribution of the lower boundary of the one-sided confidence intervals obtained from SRT, RT 2nd and RT. The naive randomization test yields the tightest confidence intervals; yet, this comes at the cost of losing type-I error control as demonstrated in the previous paragraph. Among RT 2nd and SRT, the confidence intervals computed via the selective randomization test are on average narrower. In particular, when both subgroups were selected, the lower boundary of the confidence interval from SRT and RT are quite similar.

\color{black}
\subsection{First-stage only experiments}\label{app:first.stage.only.test}

In Figure~\ref{fig:default_CRT1st_power} below, we provide the conditional randomization test based solely on the first stage (SRT(1st) in black) using the data generation mechanism of \Cref{fig:default} in the main text. We observe that when only one of the two subgroups is selected (panel 2 and 3), the number of acceptable treatments of the first stage is very small for sharp nulls with considerably small $\tau$. Consequently, there is a significant decay in the associated rejection probability (power).

In our simulation (\Cref{fig:default} in the main text), there is a second stage after the selection is made (based on the first-stage observations), and all feasible second-stage treatment assignments are acceptable.
Hence, all units in the second stage serve as hold-out units.
As a result, even if there are only few acceptable treatment assignments $z^*_1$ for the first stage left after conditioning on the observed selection $S=S(Z_1)$, there are always sufficiently many treatment assignments $z^*_2$ in the second stage to draw inference. In other words, since we use all the information that is left after conditioning on $S$, which was named `data carving' by \textcite{fithian_optimal_2017}, we can use the full information from the second stage for inference. Therefore, we do not observe the power issue in the first-stage-only analysis.

\begin{figure}[htb]
  \centering
  \includegraphics[scale=0.8]{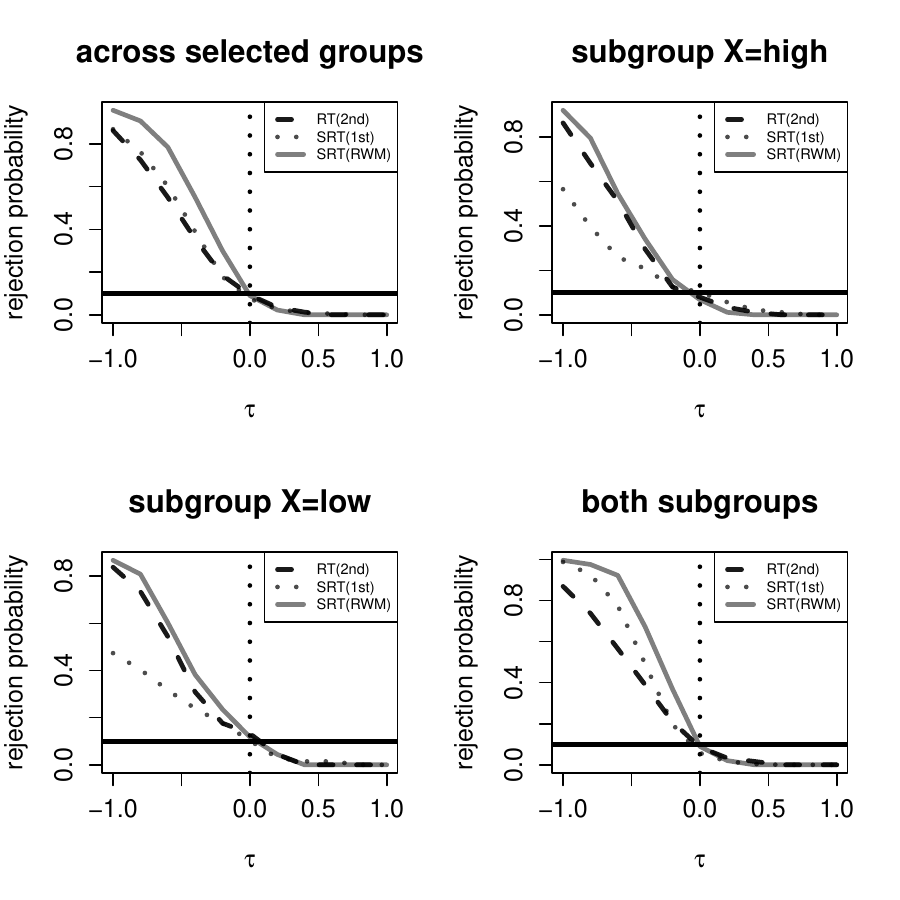}
  \caption{\color{black}Power comparison of different conditional randomization tests. RT(2nd) uses only the second stage, SRT(1st) conditions on the first stage, and SRT(RWM) implements the proposed selective randomization test. When only one subgroup is selected (subgroup $X=$ high or low), SRT(1st) suffers from low power because very few first-stage assignments remain acceptable under the sharp null $Y_i(1) - Y_i(0) = \tau$, $\tau < 0$. In contrast, SRT(RWM) and RT(2nd) maintain substantially higher power thanks to a large number of acceptable alternative treatment assignments of the second stage.}
  \label{fig:default_CRT1st_power}
    \end{figure}

\clearpage

\printbibliography

@article{zhang2025powerful,
  title={Powerful randomization tests for subgroup analysis},
  author={Zhang, Yao and Gao, Zijun},
  journal={arXiv: 2504.21572},
  year={2025}
}

@article{andrews2024inference,
  title={Inference on winners},
  author={Andrews, Isaiah and Kitagawa, Toru and McCloskey, Adam},
  journal={The Quarterly Journal of Economics},
  volume={139},
  number={1},
  pages={305--358},
  year={2024},
  publisher={Oxford University Press}
}

@article{garthwaite1996confidence,
  title={Confidence intervals from randomization tests},
  author={Garthwaite, Paul H},
  journal={Biometrics},
  volume={52},
  pages={1387--1393},
  year={1996},
  publisher={JSTOR}
}

@book{edgington2007randomization,
  title={Randomization tests},
  author={Edgington, Eugene and Onghena, Patrick},
  year={2007},
  publisher={Chapman and Hall/CRC}
}

@article{ivanova2022randomization,
  title={Randomization tests in clinical trials with multiple imputation for handling missing data},
  author={Ivanova, Anastasia and Lederman, Seth and Stark, Philip B and Sullivan, Gregory and Vaughn, Ben},
  journal={Journal of Biopharmaceutical Statistics},
  volume={32},
  number={3},
  pages={441--449},
  year={2022},
  publisher={Taylor \& Francis}
}

@article{heussen2023randomization,
  title={Randomization-based inference for clinical trials with missing outcome data},
  author={Heussen, Nicole and Hilgers, Ralf-Dieter and Rosenberger, William F and Tan, Xiao and Uschner, Diane},
  journal={Statistics in Biopharmaceutical Research},
  pages={1--17},
  year={2023},
  publisher={Taylor \& Francis}
}

@article{friede2012conditional,
  title={A conditional error function approach for subgroup selection in adaptive clinical trials},
  author={Friede, Tim and Parsons, N and Stallard, Nigel},
  journal={Statistics in medicine},
  volume={31},
  number={30},
  pages={4309--4320},
  year={2012},
  publisher={Wiley Online Library}
}

@article{stallard2014adaptive,
  title={Adaptive designs for confirmatory clinical trials with subgroup selection},
  author={Stallard, Nigel and Hamborg, Thomas and Parsons, Nicholas and Friede, Tim},
  journal={Journal of biopharmaceutical statistics},
  volume={24},
  number={1},
  pages={168--187},
  year={2014},
  publisher={Taylor \& Francis}
}

@article{lai2019adaptive,
  title={Adaptive enrichment designs for confirmatory trials},
  author={Lai, Tze Leung and Lavori, Philip W and Tsang, Ka Wai},
  journal={Statistics in medicine},
  volume={38},
  number={4},
  pages={613--624},
  year={2019},
  publisher={Wiley Online Library}
}

@article{rosenblum2016group,
  title={Group sequential designs with prospectively planned rules for subpopulation enrichment},
  author={Rosenblum, Michael and Luber, Brandon and Thompson, Richard E and Hanley, Daniel},
  journal={Statistics in medicine},
  volume={35},
  number={21},
  pages={3776--3791},
  year={2016},
  publisher={Wiley Online Library}
}

@misc{sprint2016systolic,
    title = {Systolic Blood Pressure Intervention Trial (SPRINT)},
    howpublished = {\url{https://www.clinicaltrials.gov/}},
    author = {{National Heart, Lung, and Blood Institute (NHLBI)}},
    year = {2016},
    note = {Identifier: NCT01206062},
}

@article{gao2021assessment,
	title = {Assessment of heterogeneous treatment effect estimation accuracy via matching},
	volume = {40},
	number = {17},
	journal = {Statistics in Medicine},
	author = {Gao, Zijun and Hastie, Trevor and Tibshirani, Robert},
	year = {2021},
	pages = {3990--4013}
}

@article{sherlock2010random,
  title={The random walk Metropolis: linking theory and practice through a case study},
  author={Sherlock, Chris and Fearnhead, Paul and Roberts, Gareth O},
  journal={Statistical Science},
  volume={25},
  number={2},
  pages={172--190},
  year={2010},
}

@book{gilks1996markov,
  title={Markov chain Monte Carlo in practice},
  author={Gilks, Walter R and Richardson, Sylvia and Spiegelhalter, David},
  year={1996},
  publisher={Chapman
and Hall, London}
}

@book{levin2017markov,
  title={Markov chains and mixing times},
  author={Levin, David A and Peres, Yuval},
  year={2017},
  publisher={American Mathematical Society}
}

@article{metropolis1953equation,
  title={Equation of state calculations by fast computing machines},
  author={Metropolis, Nicholas and Rosenbluth, Arianna W and Rosenbluth, Marshall N and Teller, Augusta H and Teller, Edward},
  journal={The journal of chemical physics},
  volume={21},
  number={6},
  pages={1087--1092},
  year={1953},
  publisher={American Institute of Physics}
}

@article{simon2011using,
  title={Using randomization tests to preserve type I error with response adaptive and covariate adaptive randomization},
  author={Simon, Richard and Simon, Noah Robin},
  journal={Statistics \& probability letters},
  volume={81},
  number={7},
  pages={767--772},
  year={2011},
  publisher={Elsevier}
}

@article{nair2023weightedmc,
	title = {Randomization {Tests} for {Adaptively} {Collected} {Data}},
        journal = {arXiv: 2301.0536},
	author = {Nair, Yash and Janson, Lucas},
	year = {2023}
}

@article{fithian_optimal_2017,
	title = {Optimal {Inference} {After} {Model} {Selection}},
	author = {Fithian, William and Sun, Dennis and Taylor, Jonathan},
	year = {2017},
	journal = {arXiv: 1410.2597},
}

@article{zhang_2023_randomization_test,
	title = {What is a {Randomization} {Test}?},
	volume = {0},
	number = {0},
	journal = {Journal of the American Statistical Association},
	author = {Zhang, Yao and Zhao, Qingyuan},
	year = {2023},
	pages = {1--15},
}

@article{caughey_randomisation_2023,
	title = {Randomisation inference beyond the sharp null: bounded null hypotheses and quantiles of individual treatment effects},
	volume = {85},
	shorttitle = {Randomisation inference beyond the sharp null},
	number = {5},
	journal = {Journal of the Royal Statistical Society Series B: Statistical Methodology},
	author = {Caughey, Devin and Dafoe, Allan and Li, Xinran and Miratrix, Luke},
	year = {2023},
	pages = {1471--1491}
}

@article{pallmann_adaptive_2018,
	title = {Adaptive designs in clinical trials: why use them, and how to run and report them},
	volume = {16},
	shorttitle = {Adaptive designs in clinical trials},
	number = {1},
	journal = {BMC Medicine},
	author = {Pallmann, Philip and Bedding, Alun W. and Choodari-Oskooei, Babak and Dimairo, Munyaradzi and Flight, Laura and Hampson, Lisa V. and Holmes, Jane and Mander, Adrian P. and Odondi, Lang’o and Sydes, Matthew R. and Villar, Sofía S. and Wason, James M. S. and Weir, Christopher J. and Wheeler, Graham M. and Yap, Christina and Jaki, Thomas},
	year = {2018},
	pages = {29}
}

@article{burnett_adding_2020,
	title = {Adding flexibility to clinical trial designs: an example-based guide to the practical use of adaptive designs},
	volume = {18},
    shorttitle = {Adding flexibility to clinical trial designs},
	number = {1},
	journal = {BMC Medicine},
	author = {Burnett, Thomas and Mozgunov, Pavel and Pallmann, Philip and Villar, Sofia S. and Wheeler, Graham M. and Jaki, Thomas},
	year = {2020},
	pages = {352},
}

@book{armitage_sequential_1960,
  title={Sequential medical trials},
  author={Armitage, P.M.A.},
  journal={Sequential {Medical} {Trials}},
  year={1960},
  publisher={Blackwell Scientific Publications}
}

@book{rosenberger_randomization_2015,
	title = {Randomization in {Clinical} {Trials}: {Theory} and {Practice}},
	shorttitle = {Randomization in {Clinical} {Trials}},
	publisher = {John Wiley \& Sons},
	author = {Rosenberger, William F. and Lachin, John M.},
	year = {2015}
}

@article{pocock_group_1977,
	title = {Group sequential methods in the design and analysis of clinical trials},
	volume = {64},
	number = {2},
	journal = {Biometrika},
	author = {Pocock, S. J.},
	year = {1977},
	pages = {191--199}
}

@article{anscombe_sequential_1963,
	title = {Sequential {Medical} {Trials}},
	volume = {58},
	number = {302},
	journal = {Journal of the American Statistical Association},
	author = {Anscombe, F. J.},
	year = {1963},
	pages = {365--383}
}

@article{lin_inference_2021,
	title = {Inference for a two-stage enrichment design},
	volume = {49},
	number = {5},
	journal = {The Annals of Statistics},
	author = {Lin, Zhantao and Flournoy, Nancy and Rosenberger, William F.},
	year = {2021}
}

@article{frieri_design_2022,
	title = {Design considerations for two-stage enrichment clinical trials},
	volume = {79},
	number = {3},
	journal = {Biometrics},
	author = {Frieri, Rosamarie and Rosenberger, William Fisher and Flournoy, Nancy and Lin, Zhantao},
	year = {2023},
	pages = {2565--2576}
}

@article{spencer_adaptive_2016,
	title = {An adaptive design for updating the threshold value of a continuous biomarker: {An} {Adaptive} {Design} for {Updating} the {Threshold} of a {Biomarker}},
	volume = {35},
	shorttitle = {An adaptive design for updating the threshold value of a continuous biomarker},
	number = {27},
	journal = {Statistics in Medicine},
	author = {Spencer, Amy V. and Harbron, Chris and Mander, Adrian and Wason, James and Peers, Ian},
	year = {2016},
	pages = {4909--4923}
}

@article{stallard_adaptive_2023,
	title = {Adaptive enrichment designs with a continuous biomarker},
	volume = {79},
	number = {1},
	journal = {Biometrics},
	author = {Stallard, Nigel},
	year = {2023},
	pages = {9--19}
}

@article{robertson_point_2023,
	title = {Point estimation for adaptive trial designs {I}: {A} methodological review},
	volume = {42},
	shorttitle = {Point estimation for adaptive trial designs {I}},
	number = {2},
	journal = {Statistics in Medicine},
	author = {Robertson, David S. and Choodari-Oskooei, Babak and Dimairo, Munya and Flight, Laura and Pallmann, Philip and Jaki, Thomas},
	year = {2023},
	pages = {122--145}
}

@article{muller_adaptive_2001,
	title = {Adaptive {Group} {Sequential} {Designs} for {Clinical} {Trials}: {Combining} the {Advantages} of {Adaptive} and of {Classical} {Group} {Sequential} {Approaches}},
	volume = {57},
	shorttitle = {Adaptive {Group} {Sequential} {Designs} for {Clinical} {Trials}},
	number = {3},
	journal = {Biometrics},
	author = {M\"uller, Hans-Helge and Sch\"afer, Helmut},
	year = {2001},
	pages = {886--891}
}

@article{pocock_interim_1982,
	title = {Interim {Analyses} for {Randomized} {Clinical} {Trials}: {The} {Group} {Sequential} {Approach}},
	volume = {38},
	shorttitle = {Interim {Analyses} for {Randomized} {Clinical} {Trials}},
	number = {1},
	journal = {Biometrics},
	author = {Pocock, Stuart J.},
	year = {1982},
	pages = {153}
}

@book{fisher_design_1935,
	address = {Edinburgh},
	title = {The design of experiments},
	publisher = {Oliver \& Boyd},
	author = {Fisher, R. A.},
	year = {1935}
}

@article{wang_randomization_2020,
	title = {Randomization tests for multiarmed randomized clinical trials},
	volume = {39},
	issn = {0277-6715, 1097-0258},
	number = {4},
	journal = {Statistics in Medicine},
	author = {Wang, Yanying and Rosenberger, William F. and Uschner, Diane},
	year = {2020},
	pages = {494--509}
}

@article{ji_randomization_2017,
	title = {Randomization inference for stepped-wedge cluster-randomized trials: {An} application to community-based health insurance},
	volume = {11},
	issn = {1932-6157, 1941-7330},
	number = {1},
	journal = {The Annals of Applied Statistics},
	author = {Ji, Xinyao and Fink, Gunther and Robyn, Paul Jacob and Small, Dylan S.},
	year = {2017},
	pages = {1--20}
}

@article{rosenberger_randomization_2019,
	title = {Randomization: {The} forgotten component of the randomized clinical trial},
	volume = {38},
	issn = {1097-0258},
	number = {1},
	journal = {Statistics in Medicine},
	author = {Rosenberger, William F. and Uschner, Diane and Wang, Yanying},
	year = {2019},
	pages = {1--12}
}

@article{thompson_likelihood_1933,
	title = {On the likelihood that one unknown probability exceeds another in view of the evidence of two samples},
	volume = {25},
	number = {3-4},
	journal = {Biometrika},
	author = {Thompson, William R.},
	year = {1933},
	pages = {285--294}
}

@article{hadad_confidence_2021,
	title = {Confidence intervals for policy evaluation in adaptive experiments},
	volume = {118},
	number = {15},
	journal = {Proceedings of the National Academy of Sciences},
	author = {Hadad, Vitor and Hirshberg, David A. and Zhan, Ruohan and Wager, Stefan and Athey, Susan},
	year = {2021},
	pages = {e2014602118}
}

@inproceedings{zhang_inference_2020,
	title = {Inference for {Batched} {Bandits}},
	volume = {33},
	booktitle = {Advances in {Neural} {Information} {Processing} {Systems}},
	publisher = {Curran Associates, Inc.},
	author = {Zhang, Kelly and Janson, Lucas and Murphy, Susan},
	year = {2020},
	pages = {9818--9829}
}

@article{howard_time-uniform_2021,
	title = {Time-uniform, nonparametric, nonasymptotic confidence sequences},
	volume = {49},
	number = {2},
	journal = {The Annals of Statistics},
	author = {Howard, Steven R. and Ramdas, Aaditya and McAuliffe, Jon and Sekhon, Jasjeet},
	year = {2021}
}

@article{dawid1979conditional,
  title={Conditional independence in statistical theory},
  author={Dawid, A Philip},
  journal={Journal of the Royal Statistical Society Series B: Statistical Methodology},
  volume={41},
  number={1},
  pages={1--15},
  year={1979},
  publisher={Oxford University Press}
}

@book{pearlProbabilisticReasoningIntelligent1988,
  title = {Probabilistic Reasoning in Intelligent Systems: {{Networks}} of Plausible Inference},
  author = {Pearl, Judea},
  year = {1988},
  publisher = {Morgan Kaufmann Publishers Inc.}
}

@article{lee_exact_2016,
	title = {Exact post-selection inference, with application to the lasso},
	volume = {44},
	number = {3},
	journal = {The Annals of Statistics},
	author = {Lee, Jason D. and Sun, Dennis L. and Sun, Yuekai and Taylor, Jonathan E.},
	year = {2016}
}

@article{kuchibhotla_post-selection_2022,
	title = {Post-{Selection} {Inference}},
	volume = {9},
	number = {1},
	urldate = {2024-01-10},
	journal = {Annual Review of Statistics and Its Application},
	author = {Kuchibhotla, Arun K. and Kolassa, John E. and Kuffner, Todd A.},
	year = {2022},
	pages = {505--527}
}

@article{basse2019randomization,
  title = {Randomization Tests of Causal Effects under Interference},
  author = {Basse, {\relax GW} and Feller, A and Toulis, P},
  year = {2019},
  journal = {Biometrika},
  volume = {106},
  number = {2},
  pages = {487--494},
  publisher = {Oxford University Press}
}

@article{cox_note_1975,
	title = {A note on data-splitting for the evaluation of significance levels},
	volume = {62},
	number = {2},
	journal = {Biometrika},
	author = {Cox, D. R.},
	year = {1975},
	pages = {441--444}
}

@article{edwards_model_1999,
	title = {On model prespecification in confirmatory randomized studies},
	volume = {18},
	number = {7},
	journal = {Statistics in Medicine},
	author = {Edwards, David},
	year = {1999},
	pages = {771--785}
}

@article{box1957,
 journal = {Biometrics},
 number = {2},
 pages = {238--246},
 author = {George E. P. Box},
 publisher = {Wiley, International Biometric Society},
 title = {Abstracts},
 volume = {13},
 year = {1957}
}

@inproceedings{tukey1961,
    author = "Tukey, John W.",
    title = "Statistical and quantitative methodology",
    booktitle = "Trends in Social Science",
    pages = "84--136",
    publisher = "Philosophical Library",
    address = "New York",
    year = "1961"
}

@article{offer-westort_adaptive_2021,
	title = {Adaptive {Experimental} {Design}: {Prospects} and {Applications} in {Political} {Science}},
	volume = {65},
	shorttitle = {Adaptive {Experimental} {Design}},
	number = {4},
	journal = {American Journal of Political Science},
	author = {Offer-Westort, Molly and Coppock, Alexander and Green, Donald P.},
	year = {2021},
	pages = {826--844}
}

@article{kasy_adaptive_2021,
	title = {Adaptive {Treatment} {Assignment} in {Experiments} for {Policy} {Choice}},
	volume = {89},
	number = {1},
	journal = {Econometrica},
	author = {Kasy, Maximilian and Sautmann, Anja},
	year = {2021},
	pages = {113--132}
}

@inproceedings{deshpande_accurate_2018,
	title = {Accurate {Inference} for {Adaptive} {Linear} {Models}},
	booktitle = {Proceedings of the 35th {International} {Conference} on {Machine} {Learning}},
	publisher = {PMLR},
	author = {Deshpande, Yash and Mackey, Lester and Syrgkanis, Vasilis and Taddy, Matt},
	year = {2018},
	pages = {1194--1203}
}

@inproceedings{zhang_statistical_2021,
	title = {Statistical {Inference} with {M}-{Estimators} on {Adaptively} {Collected} {Data}},
	volume = {34},
	booktitle = {Advances in {Neural} {Information} {Processing} {Systems}},
	publisher = {Curran Associates, Inc.},
	author = {Zhang, Kelly and Janson, Lucas and Murphy, Susan},
	year = {2021},
    pages = {7460 -- 7471}
}

@article{jr_estimates_1963,
	title = {Estimates of {Location} {Based} on {Rank} {Tests}},
	volume = {34},
	number = {2},
	journal = {The Annals of Mathematical Statistics},
	author = {Hodges, J. L. and Lehmann, E. L.},
    publisher = {Institute of Mathematical Statistics},
	year = {1963},
	pages = {598--611}
}

@article{rosenbaum_hodges-lehmann_1993,
	title = {Hodges-{Lehmann} {Point} {Estimates} of {Treatment} {Effect} in {Observational} {Studies}},
	volume = {88},
	number = {424},
	journal = {Journal of the American Statistical Association},
	author = {Rosenbaum, Paul R.},
    publisher = {Taylor \& Francis},
	year = {1993},
	pages = {1250--1253}
}

@article{marston_predicting_2020,
	title = {Predicting {Benefit} {From} {Evolocumab} {Therapy} in {Patients} {With} {Atherosclerotic} {Disease} {Using} a {Genetic} {Risk} {Score}},
	volume = {141},
	number = {8},
	journal = {Circulation},
	author = {Marston, Nicholas A. and Kamanu, Frederick K. and Nordio, Francesco and Gurmu, Yared and Roselli, Carolina and Sever, Peter S. and Pedersen, Terje R. and Keech, Anthony C. and Wang, Huei and Lira Pineda, Armando and Giugliano, Robert P. and Lubitz, Steven A. and Ellinor, Patrick T. and Sabatine, Marc S. and Ruff, Christian T.},
	year = {2020},
	pages = {616--623}
}

@article{faseru_changing_2017,
	title = {Changing the default for tobacco-cessation treatment in an inpatient setting: study protocol of a randomized controlled trial},
	volume = {18},
	shorttitle = {Changing the default for tobacco-cessation treatment in an inpatient setting},
	number = {1},
	journal = {Trials},
	author = {Faseru, Babalola and Ellerbeck, Edward F. and Catley, Delwyn and Gajewski, Byron J. and Scheuermann, Taneisha S. and Shireman, Theresa I. and Mussulman, Laura M. and Nazir, Niaman and Bush, Terry and Richter, Kimber P.},
	year = {2017},
	pages = {379}
}

@article{rosenberger_covariate-adjusted_2001,
	title = {Covariate-{Adjusted} {Response}-{Adaptive} {Designs} for {Binary} {Response}},
	volume = {11},
	number = {4},
	journal = {Journal of Biopharmaceutical Statistics},
	author = {Rosenberger, William F. and Vidyashankar, A. N. and Agarwal, Deepak K.},
	year = {2001},
	pages = {227--236}
}

@article{ho_efficacy_2012,
	title = {Efficacy and tolerability of rizatriptan in pediatric migraineurs: {Results} from a randomized, double-blind, placebo-controlled trial using a novel adaptive enrichment design},
	volume = {32},
	shorttitle = {Efficacy and tolerability of rizatriptan in pediatric migraineurs},
	number = {10},
	journal = {Cephalalgia},
	author = {Ho, Tony W and Pearlman, Eric and Lewis, Donald and Hämäläinen, Mirja and Connor, Kathryn and Michelson, David and Zhang, Ying and Assaid, Christopher and Mozley, Lyn Harper and Strickler, Nancy and Bachman, Robert and Mahoney, Erin and Lines, Christopher and Hewitt, David J},
	year = {2012},
	pages = {750--765}
}

@article{magnusson_group_2013,
	title = {Group sequential enrichment design incorporating subgroup selection},
	volume = {32},
	number = {16},
	journal = {Statistics in Medicine},
	author = {Magnusson, Baldur P. and Turnbull, Bruce W.},
	year = {2013},
	pages = {2695--2714}
}

@article{magirr_generalized_2012,
	title = {A generalized {Dunnett} test for multi-arm multi-stage clinical studies with treatment selection},
	volume = {99},
	number = {2},
	journal = {Biometrika},
	author = {Magirr, D. and Jaki, T. and Whitehead, J.},
	year = {2012},
	pages = {494--501}
}

@article{sydes_issues_2009,
	title = {Issues in applying multi-arm multi-stage methodology to a clinical trial in prostate cancer: the {MRC} {STAMPEDE} trial},
	volume = {10},
	shorttitle = {Issues in applying multi-arm multi-stage methodology to a clinical trial in prostate cancer},
	number = {1},
	journal = {Trials},
	author = {Sydes, Matthew R. and Parmar, Mahesh KB and James, Nicholas D. and Clarke, Noel W. and Dearnaley, David P. and Mason, Malcolm D. and Morgan, Rachel C. and Sanders, Karen and Royston, Patrick},
	year = {2009},
	pages = {39}
}

@article{pitman_significance_1937,
	title = {Significance {Tests} {Which} {May} be {Applied} to {Samples} {From} any {Populations}},
	volume = {4},
	number = {1},
	journal = {Supplement to the Journal of the Royal Statistical Society},
	author = {Pitman, E. J. G.},
	year = {1937},
	pages = {119--130}
}

@incollection{duflo_chapter_2007,
	title = {Chapter 61 {Using} {Randomization} in {Development} {Economics} {Research}: {A} {Toolkit}},
	volume = {4},
	shorttitle = {Chapter 61 {Using} {Randomization} in {Development} {Economics} {Research}},
	booktitle = {Handbook of {Development} {Economics}},
	author = {Duflo, Esther and Glennerster, Rachel and Kremer, Michael},
	year = {2007},
	pages = {3895--3962}
}

@article{bugni_inference_2018,
	title = {Inference {Under} {Covariate}-{Adaptive} {Randomization}},
	volume = {113},
	number = {524},
	journal = {Journal of the American Statistical Association},
	author = {Bugni, Federico A. and Canay, Ivan A. and Shaikh, Azeem M.},
	year = {2018},
	pages = {1784--1796}
}

@book{scheffeAnalysisVariance1959,
  title = {The Analysis of Variance},
  author = {Scheff{\'e}, Henry},
  year = {1959},
  publisher = {John Wiley and Sons}
}

@article{tukeyComparingIndividualMeans1949,
  title = {Comparing Individual Means in the Analysis of Variance},
  author = {Tukey, John W.},
  year = {1949},
  journal = {Biometrics},
  volume = {5},
  number = {2},
  pages = {99--114},
  publisher = {[Wiley, International Biometric Society]},
}

@article{benjaminiControllingFalseDiscovery1995,
  title = {Controlling the False Discovery Rate: A Practical and Powerful Approach to Multiple Testing},
  shorttitle = {Controlling the {{False Discovery Rate}}},
  author = {Benjamini, Yoav and Hochberg, Yosef},
  year = {1995},
  journal = {Journal of the Royal Statistical Society. Series B (Methodological)},
  volume = {57},
  number = {1},
  pages = {289--300},
  publisher = {[Royal Statistical Society, Wiley]}
}

@article{benjaminiFalseDiscoveryRate2005a,
  title = {False Discovery Rate-Adjusted Multiple Confidence Intervals for Selected Parameters},
  author = {Benjamini, Yoav and Yekutieli, Daniel},
  year = {2005},
  journal = {Journal of the American Statistical Association},
  volume = {100},
  number = {469},
  pages = {71--81},
  publisher = {Taylor \& Francis}
}

@article{benjaminiSelectiveInferenceSilent2020,
  title = {Selective {{Inference}}: {{The Silent Killer}} of {{Replicability}}},
  shorttitle = {Selective {{Inference}}},
  author = {Benjamini, Yoav},
  year = {2020},
  journal = {Harvard Data Science Review},
  volume = {2},
  number = {4}
}

@article{berkValidPostselectionInference2013,
  title = {Valid Post-Selection Inference},
  author = {Berk, Richard and Brown, Lawrence and Buja, Andreas and Zhang, Kai and Zhao, Linda},
  year = {2013},
  journal = {The Annals of Statistics},
  volume = {41},
  number = {2},
  pages = {802--837},
  publisher = {Institute of Mathematical Statistics}
}

@article{morgan_rerandomization_2012,
	title = {Rerandomization to improve covariate balance in experiments},
	volume = {40},
	number = {2},
	journal = {The Annals of Statistics},
	author = {Morgan, Kari Lock and Rubin, Donald B.},
	year = {2012},
	pages = {1263--1282}
}

@article{cohen_gaussian_2022,
	title = {Gaussian {Prepivoting} for {Finite} {Population} {Causal} {Inference}},
	volume = {84},
	number = {2},
	journal = {Journal of the Royal Statistical Society Series B: Statistical Methodology},
	author = {Cohen, Peter L. and Fogarty, Colin B.},
	year = {2022},
	pages = {295--320}
}

@article{wu_randomization_2021,
	title = {Randomization {Tests} for {Weak} {Null} {Hypotheses} in {Randomized} {Experiments}},
	volume = {116},
	number = {536},
	journal = {Journal of the American Statistical Association},
	author = {Wu, Jason and Ding, Peng},
	year = {2021},
	pages = {1898--1913}
}

@article{ding_paradox_2017,
	title = {A {Paradox} from {Randomization}-{Based} {Causal} {Inference}},
	volume = {32},
	number = {3},
	journal = {Statistical Science},
	author = {Ding, Peng},
	year = {2017},
	pages = {331--345}
}

@book{imbens2015causal,
  title={Causal Inference for Statistics, Social, and Biomedical Sciences: An Introduction},
  author={Imbens, Guido W and Rubin, Donald B},
  year={2015},
  publisher={Cambridge University Press}
}

@article{almirall_introduction_2014,
	title = {Introduction to {SMART} designs for the development of adaptive interventions: with application to weight loss research},
	volume = {4},
	shorttitle = {Introduction to {SMART} designs for the development of adaptive interventions},number = {3},
	journal = {Translational Behavioral Medicine},
	author = {Almirall, Daniel and Nahum-Shani, Inbal and Sherwood, Nancy E. and Murphy, Susan A.},
	year = {2014},
	pages = {260--274}
}

@article{panigrahi_carving_2023,
	title = {Carving model-free inference},
	volume = {51},
	number = {6},
	journal = {The Annals of Statistics},
	author = {Panigrahi, Snigdha},
	year = {2023},
	pages = {2318--2341}
}

@book{hernan_causal_2024,
	title = {Causal {Inference}: {What} {If}},
	shorttitle = {Causal {Inference}},
	publisher = {Boca Raton: Chapman \& Hall/CRC},
	author = {Hern\'an, Miguel A. and Robins, James M.},
	year = {2024}
}

@article{efron_forcing_1971,
	title = {Forcing a sequential experiment to be balanced},
	volume = {58},
	number = {3},
	journal = {Biometrika},
	author = {Efron, Bradley},
	year = {1971},
	pages = {403--417}
}

@article{ye_toward_2023,
	title = {Toward {Better} {Practice} of {Covariate} {Adjustment} in {Analyzing} {Randomized} {Clinical} {Trials}},
	volume = {118},
	number = {544},
	journal = {Journal of the American Statistical Association},
	author = {Ye, Ting and Shao, Jun and Yi, Yanyao and Zhao, Qingyuan},
	year = {2023},
	pages = {2370--2382}
}

@article{besag_sequential_1991,
	title = {Sequential {Monte} {Carlo} p-values},
	volume = {78},
	number = {2},
	journal = {Biometrika},
	author = {Besag, Julian and Clifford, Peter},
	year = {1991},
	pages = {301--304}
}

@article{fay_using_2007,
	title = {On {Using} {Truncated} {Sequential} {Probability} {Ratio} {Test} {Boundaries} for {Monte} {Carlo} {Implementation} of {Hypothesis} {Tests}},
	volume = {16},
	number = {4},
	journal = {Journal of Computational and Graphical Statistics},
	author = {Fay, Michael P and Kim, Hyune-Ju and Hachey, Mark},
	year = {2007},
	pages = {946--967}
}

@article{silva_optimal_2013,
	title = {Optimal generalized truncated sequential {Monte} {Carlo} test},
	volume = {121},
	journal = {Journal of Multivariate Analysis},
	author = {Silva, Ivair R. and Assunção, Renato M.},
	year = {2013},
	pages = {33--49},
}

@article{fischer_sequential_2024,
	title = {Sequential {Monte}-{Carlo} testing by betting},
    journal = {arXiv: 2401.07365},
	author = {Fischer, Lasse and Ramdas, Aaditya},
	year = {2024}
}

@article{pimentel2024,
      title={Re-evaluating the impact of hormone replacement therapy on heart disease using match-adaptive randomization inference},
      author={Samuel D. Pimentel and Ruoqi Yu},
      year={2024},
    journal = {arXiv: 2403.01330}
}

@article{pimentel_covariate-adaptive_2024,
	title = {Covariate-adaptive randomization inference in matched designs},
	journal = {Journal of the Royal Statistical Society Series B: Statistical Methodology},
	author = {Pimentel, Samuel D and Huang, Yaxuan},
	year = {2024},
	pages = {qkae033},
}

@article{gao_selective_2025,
	title = {Selective randomization inference for subgroup effects with continuous biomarkers},
        journal = {arXiv: 2504.19380},
	author = {Gao, Zijun},
	year = {2025}
}

@article{zheng_multi-center_2008,
	title = {{Multi}-{Center} {Clinical} {Trials}: {Randomization} {and} {Ancillary} {Statistics}},
	volume = {2},
	number = {2},
	journal = {The Annals of Applied Statistics},
	author = {Zheng, Lu and Zelen, Marvin},
	year = {2008},
	pages = {582--600}
}

@article{hirano_asymptotic_2025,
	title = {Asymptotic {Representations} for {Sequential} {Decisions}, {Adaptive} {Experiments}, and {Batched} {Bandits}},
        journal = {arXiv: 2302.03117},
	author = {Hirano, Keisuke and Porter, Jack R.},
	year = {2025},
}

@article{chen_optimal_2023,
	title = {Optimal {Conditional} {Inference} in {Adaptive} {Experiments}},
        journal = {arXiv: 2309.12162},
	author = {Chen, Jiafeng and Andrews, Isaiah},
	year = {2023}
}

@article{howes_markov_2024,
	title = {Markov {Chain} {Monte} {Carlo} {Significance} {Tests}},
        journal = {arXiv: 2310.04924},
	author = {Howes, Michael},
	year = {2024}
}

\end{document}